\documentclass[lettersize,journal]{IEEEtran}
\usepackage{amsfonts}
\usepackage{algorithmic}
\usepackage{algorithm}
\usepackage{array}
\usepackage[caption=false,font=normalsize,labelfont=sf,textfont=sf]{subfig}
\usepackage{textcomp}
\usepackage{stfloats}
\usepackage{url}
\usepackage{verbatim}
\usepackage{graphicx}
\usepackage{cite}
\usepackage{amsmath, amsthm, amssymb}
\usepackage{multirow}
\usepackage{booktabs}
\usepackage[utf8]{inputenc}

\newtheorem{proposition}{Proposition}

\begin{document}

\title{Attribute-Prompted Kernel Hashing for Unsupervised \\ Data-Efficient Cross-Modal Retrieval}

\author{Runhao~Li, Xiaoxu~Ma, Zhenyu~Weng, Yue~Zhang, Guibo~Luo, Huiping~Zhuang, Zhiping~Lin,~\IEEEmembership{Life~Senior~Member,~IEEE}, and Yap-Peng~Tan,~\IEEEmembership{Fellow,~IEEE}

\thanks{
Runhao~Li, Zhiping~Lin, and Yap-Peng~Tan are with the School of Electrical and Electronic Engineering, Nanyang Technological University, Singapore. Yap-Peng~Tan is also with VinUniversity, Hanoi, Vietnam (e-mail: runhao001@e.ntu.edu.sg; ezplin@ntu.edu.sg; eyptan@ntu.edu.sg / yp.t@vinuni.edu.vn).}

\thanks{
Xiaoxu~Ma, Zhenyu~Weng, and Huiping~Zhuang are with Shien-Ming Wu School of Intelligent Engineering, South China University of Technology, Guangzhou, China (e-mail: \mbox{xiaoxuma0723@gmail.com}; \mbox{wzytumbler@gmail.com}; \mbox{hpzhuang@scut.edu.cn}). (Corresponding Author: Zhenyu~Weng.)}

\thanks{
Yue~Zhang is with the College of Computer and Information Engineering, Henan Normal University, China (e-mail: 2023121@htu.edu.cn)}

\thanks{
Guibo~Luo is with Guangdong Provincial Key Laboratory of Ultra High Definition Immersive Media Technology, Peking University Shenzhen Graduate School, China (e-mail: luogb@pku.edu.cn)}}


\maketitle

\begin{abstract}
Unsupervised cross-modal hashing enables efficient retrieval of semantically related instances across different modalities without requiring manual semantic annotation.
However, existing unsupervised methods rely heavily on large-scale image-text pairs. Collecting such data can be costly, particularly in scenarios where well-aligned pairs are scarce due to privacy and specialized constraints.
More critically, existing methods tend to overfit to seen training data, restricting their generalization performance on unseen categories that the constrained training data cannot cover. To address these limitations, we propose Attribute-Prompted Kernel Hashing (APKH), a novel data-efficient approach that constructs a compact, modality-aligned Hamming space driven by the generalized attribute priors of vision-language foundation models. Specifically, APKH introduces two core modules: Context-optimized Attribute Kernel Mapping (CAKM) and Kernel-Smoothed Contrastive Alignment (KSCA). CAKM formulates cross-modal alignment through hyperspherical Radial Basis Function kernel mapping, optimizing dynamic attribute kernels via prompt learning to capture modality-invariant semantics. Furthermore, KSCA extends conventional point-to-point contrastive learning by modeling limited paired data as continuous kernel distributions.
This explicit smoothing of the modality gap alleviates overfitting to sparse pairwise correlations. Extensive experiments demonstrate that APKH outperforms state-of-the-art hashing methods in the challenging cross-modal retrieval tasks from seen to unseen categories under data-constrained scenarios.
\end{abstract}


\section{Introduction}

\IEEEPARstart{T}{he} global proliferation of multimedia data, driven by social media, mobile sensors, and content-sharing networks, has necessitated advanced cross-modal retrieval systems capable of bridging the semantic gap between disparate data types, such as text and images. While cross-modal retrieval is essential for identifying semantically related instances across modalities, traditional methods based on real-valued representation learning are increasingly constrained by high storage costs and especially computational latency as datasets scale. To address these challenges, cross-modal hashing (CMH) has emerged as a vital solution, transforming high-dimensional heterogeneous data into a compact binary Hamming space. By mapping instances into binary codes, CMH facilitates efficient similarity measurement through bitwise XOR operations while simultaneously reducing storage requirements. This provides a scalable and computationally efficient framework for managing the exponential growth of large-scale multimedia databases.

The development of CMH has led to two primary trajectories: supervised learning\cite{jiang2017deep,li2024ckdh,li2023neighborhood,qin2024deep,gu2025dual,chen2025pfedlah,tu2025cross,11341891}, which leverages semantic labels to ensure high retrieval precision, and unsupervised learning\cite{kumar2011learning,liu2020joint,su2019deep,zhang2018unsupervised,creswell2018generative,hu2020creating,yu2021deep,hu2022UCCH,xia2023clip,yu2022self,xi2023unsupervised,zhu2022work,zhuo2022clip4hashing,sun2024dual,li2023clip,mingyong2023clip,10463060}, which eliminates the need for costly manual annotation by exploiting inherent co-occurrence relationships within paired data. While supervised approaches are often constrained by the scarcity of labeled data in specialized domains, unsupervised methods have gained prominence for their ability to capture complex, non-linear semantic associations directly from raw, unlabeled multi-modal inputs.
Technically, these unsupervised methods typically bridge the heterogeneity gap by constructing a dual-stream architecture consisting of two modality-specific hash functions.
The training process of existing unsupervised methods generally follows a similarity-reconstruction paradigm. First, they construct a global or batch-wise instance similarity matrix by calculating distances or mining latent correlations between deep features of co-occurring image-text pairs. Subsequently, the two modality-specific hash functions are optimized to project heterogeneous features into a unified Hamming space while ensuring that the Hamming distances between binary codes faithfully reconstruct the pre-defined similarity structure.

Despite these advances, unsupervised CMH methods still require a considerable amount of image-text pairs to effectively train the hash functions. However, these image-text pairs are often scarce and expensive to acquire, particularly when dealing with specialized or private databases where access to well-aligned image-text pairs is limited due to domain-specific constraints or privacy concerns. Recent vision-language foundation models like CLIP~\cite{radford2021learning} provide a powerful toolset to address these semantic deficiencies, as they possess an intricate high-dimensional semantic space that exhibits complex, non-linear characteristics. In the correspondence-constrained scenarios where paired data is scarce, existing frameworks that rely on modality-specific hash functions are still highly prone to overfitting when training samples are sparse. This often leads to a fractured Hamming space with the underlying modality-invariant semantic relationships poorly preserved.

To address these limitations, we propose Attribute-Prompted Kernel Hashing (APKH), a novel framework designed for data-efficient cross-modal retrieval. As shown in Fig.~\ref{fig:teaser}, unlike conventional approaches that rely on modality-specific hash functions and pointwise training for alignment, APKH reformulates the problem within a kernel-based paradigm to mitigate overfitting and enhance robustness in navigating the complex semantic manifold of CLIP.
Fig.~\ref{fig:framework} illustrates the overall architecture of the proposed APKH. The framework introduces two core modules: Context-optimized Attribute Kernel Mapping (CAKM) and Kernel-Smoothed Contrastive Alignment (KSCA).
CAKM leverages the powerful attribute priors of CLIP to establish an adaptive kernel-based alignment framework through prompt learning. It utilizes learnable context vectors to dynamically generate semantic attribute-driven kernels. This module constructs a set of modality-invariant reference attributes that act as Radial Basis Function (RBF) centers in a semantic kernel space.  This adaptive mechanism allows the model to transform the relationship between images and text into an attribute kernel response mapping, where complex correspondences become manageable and more robust.
KSCA further enhances correspondence efficiency by transitioning from discrete point-to-point learning to continuous kernel distribution alignment.  Inspired by Kernel Density Estimation~\cite{davis2011remarks} and Parzen Window theory~\cite{parzen1962estimation}, KSCA models limited paired data as local density distributions rather than isolated points and generates smoothed samples that fill the modality gap in the cross-modal manifold.  This smoothing mechanism acts as a powerful regularizer, ensuring that the hash function learns a smooth, modality-invariant Hamming space that generalizes beyond the sparse training samples. 
\begin{figure*}[]
\centering
\includegraphics[width=1.99\columnwidth]{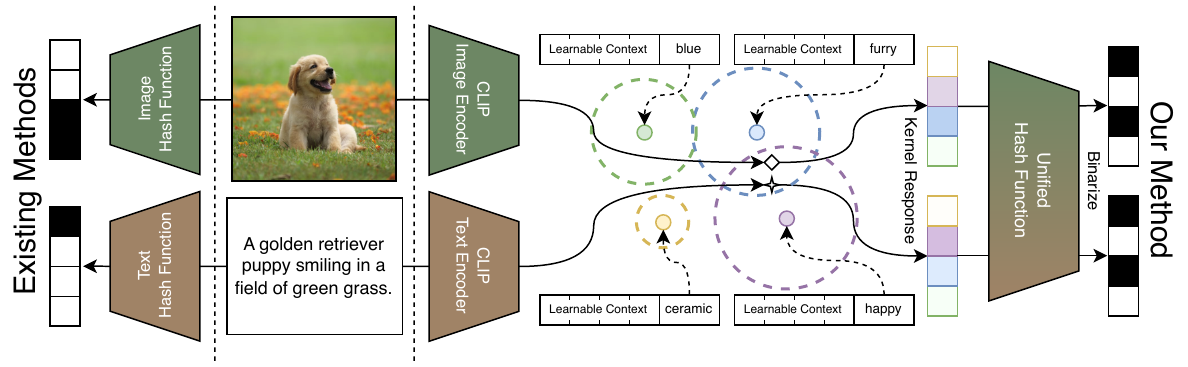} 
\caption{The inference pipelines of existing methods vs. our proposed approach. While existing methods (left) map inputs using separate modality-specific hash functions, APKH (right) calculates modality-invariant kernel responses with learnable context and generates binary codes through a unified hash function.}
\label{fig:teaser}
\end{figure*}
The main contributions of this work are summarized
as follows:
\begin{itemize}

\item We introduce APKH, a novel correspondence-efficient paradigm that utilizes the generalized attribute priors of vision-language foundation models (e.g., CLIP) to achieve robust modality alignment and effectively alleviates overfitting in data-constrained scenarios. 

\item We design CAKM, which reformulates cross-modal alignment with RBF kernel mapping. By adapting CLIP’s latent space through context-optimized attribute kernels, CAKM transforms CLIP’s multi-modal semantic manifold into modality-invariant kernel responses.

\item We develop KSCA, a statistical regularization strategy based on adaptive kernel smoothing. It explicitly fills the modality gap by modeling paired data as continuous kernel distributions, mitigating overfitting in data-constrained scenarios.  

\item Extensive evaluations on four widely-used datasets demonstrate that APKH outperforms existing hashing methods with limited training image-text pairs.

\end{itemize}

This work significantly extends our preliminary work, GNAH~\cite{11084616}. To overcome GNAH’s inherent dependency on training-constrained prototypes and modality-specific projectors, this work introduces a novel generalized data-efficient CMH framework with three major advancements: 
(1) restructuring the framework by replacing static prototypes and dual-stream projectors with a unified hash function driven by dynamic, context-optimized attribute kernels, while streamlining the complicated loss weighting scheme into a more elegant optimization objective;
(2) providing a theoretical proof of RBF-contrastive congruence to establish a mathematical foundation for the alignment mechanism;
and (3) validating the approach through more comprehensive evaluation, covering multi-bit scalability, seen-to-unseen transferability, and fine-grained ablations.

\section{Related Work}

\subsection{Unsupervised Cross-Modal Hashing}
Unsupervised cross-modal hashing (CMH) methods aim to learn hashing networks without requiring labeled data. Instead, they rely on the co-occurrence and other semantic relationships among instances, leveraging information from multiple modalities. While early methods based on hand-crafted features (e.g., CVH~\cite{kumar2011learning}, FSH~\cite{liu2017cross}) and subsequent deep learning approaches~\cite{yu2021deep,liu2020joint,su2019deep} progressively improved retrieval capabilities, they have largely been superseded by CLIP-based techniques. The current landscape of CLIP-driven hashing either employs the foundation model as a powerful backbone for feature extraction (e.g., UCMFH~\cite{xia2023clip}, SACH~\cite{yu2022self}) or uses its pre-trained knowledge to generate similarity matrices that supervise hash code learning (e.g., CDUH~\cite{xi2023unsupervised}, CLIP4Hashing~\cite{zhuo2022clip4hashing}, CAGAN~\cite{li2023clip}, CFRH~\cite{mingyong2023clip}, CDTH~\cite{10463060}).
Despite their significant advancements, a critical limitation persists: these state-of-the-art methods generally exhibit limited data efficiency. Because most of these frameworks rely on modality-specific hash functions and rigid point-to-point alignment paradigms, they require a substantial volume of well-aligned image-text pairs to adequately stabilize the cross-modal projection process. Consequently, when training data is highly constrained, these models are highly susceptible to overfitting to sparse pairwise correspondences, failing to preserve a robust semantic space.

\subsection{Data-Efficient Adaptation of Foundation Models via Prompt Learning}

As foundation models have grown in scale and prominence, prompt learning~\cite{lester2021power,zhou2022learning,yao2024tcp,li2025class,xiong2025class,li2025map} has emerged as a highly efficient alternative to costly full model fine-tuning. Originally introduced in NLP~\cite{lester2021power}, this paradigm optimizes a compact set of continuous soft prompt vectors appended to the input, flexibly steering the behavior of a frozen model without updating its massive weights. Learned dynamically via standard backpropagation, these continuous embeddings effectively replace laborious manual prompt engineering. Despite updating only a fraction of parameters, they achieve downstream performance comparable to full fine-tuning. Driven by these advantages, this parameter-efficient strategy has successfully transitioned to vision-language models~\cite{zhou2022learning,yao2024tcp,li2025class,xiong2025class}. For instance, CoOp~\cite{zhou2022learning} applies learnable soft prompts to the text encoder of a frozen CLIP model~\cite{radford2021learning}. This yields significant accuracy gains and strong cross-domain transferability with minimal training data.

\subsection{Kernel Methods}
Kernel methods have long been a fundamental technique in machine learning for analyzing non-linear patterns by implicitly mapping input data into a Reproducing Kernel Hilbert Space, where complex relationships become linearly separable~\cite{hofmann2008kernel}. In the early stages of large-scale retrieval, pioneering works such as KLSH~\cite{kulis2009kernelized} and KSH~\cite{liu2012supervised} were proposed to capture non-linear semantic similarities. Recent methods like LcKTLSH~\cite{maggu2025label} have attempted to introduce kernel mappings to cross-modal hashing. However, most existing kernel-based frameworks still rely on pre-defined kernel functions, most notably static RBF kernels, and hand-crafted anchor points tailored for early shallow features. Consequently, these static mechanisms lack the expressive flexibility to navigate the complex multi-modal semantic manifold of CLIP, nor can they interact with foundation models via prompt learning.

\begin{figure*}[]
\centering
\includegraphics[width=1.99\columnwidth]{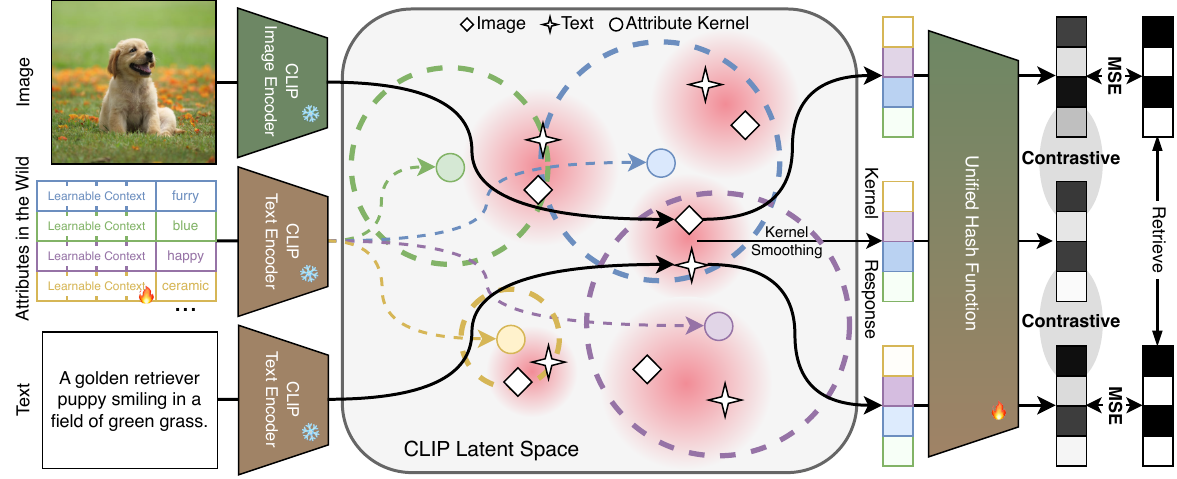} 
\caption{Framework of our proposed method. Heterogeneous inputs (images and texts) are first encoded by frozen CLIP models. Our method introduces learnable attribute prompts to construct adaptive semantic kernels. The extracted image/text representations are converted into modality-invariant kernel responses. Concurrently, a kernel smoothing mechanism is applied to bridge the semantic modality gap and mitigate overfitting. A unified hash function then binarizes these smoothed responses into hash codes, which are optimized through kernel-smoothed contrastive alignment to ensure high-quality retrieval.}
\label{fig:framework}
\end{figure*}

\section{Methodology}
\subsection{Problem Formulation}

Let $\mathcal{D}=\{\mathbf{x}_i, \mathbf{y}_i\}_{i=1}^n$ denote a cross-modal dataset with $n$ image-text pairs, where $\mathbf{x}_i \in \mathbb{R}^d$ and $\mathbf{y}_i \in \mathbb{R}^d$ represent the normalized continuous feature embeddings extracted by a frozen vision-language foundation model (e.g., CLIP) for the $i$-th image and text instances, respectively.  Our goal is to map the continuous embeddings $\mathbf{x}_i$ and $\mathbf{y}_i$ into a compact $B$-bit Hamming space. 
Specifically, we first define a non-linear kernel mapping $\Phi: \mathbb{R}^d \rightarrow \mathbb{R}^K$ that transforms the heterogeneous embeddings into shared $K$-dimensional kernel response vectors, denoted as $\mathbf{v}_i^x = \Phi(\mathbf{x}_i)$ and $\mathbf{v}_i^y = \Phi(\mathbf{y}_i)$. Subsequently, a single unified hashing network $H(\cdot): \mathbb{R}^K \rightarrow \mathbb{R}^B$ is employed to map these continuous kernel responses into a compact $B$-bit Hamming space.
The outputs of the hashing network are relaxed hash codes defined as $\mathbf{h}_i^x=H(\mathbf{v}_i^x)$ and $\mathbf{h}_i^y=H(\mathbf{v}_i^y)$ and the corresponding binary codes are defined as $\mathbf{b}_i^x = \textrm{sgn}(\mathbf{h}_i^x)$ and $\mathbf{b}_i^y = \textrm{sgn}(\mathbf{h}_i^y)$, where sgn($\cdot$) is the element-wise sign function. To achieve correspondence-efficient retrieval under data-constrained scenarios, our proposed Attribute-Prompted Kernel Hashing (APKH) framework introduces two core modules: Context-optimized Attribute Kernel Mapping and Kernel-Smoothed Contrastive Alignment.

\subsection{Context-optimized Attribute Kernel Mapping}
The Context-optimized Attribute Kernel Mapping module serves as the primary alignment engine of the APKH framework. It resolves the erratic nature of direct instance matching by establishing an adaptive attribute kernel-based alignment space. By mathematically transforming the heterogeneous relationships between disparate images and text into a unified set of modality-invariant kernel responses, CAKM effectively stabilizes the cross-modal projection mechanism.

\textbf{Attribute-Prompted Kernel Generation.} We define a set of $K$ in-the-wild attributes $\{a_k\}_{k=1}^K$ that contains semantics shared across both image and text modalities (e.g., ``happy"). To generate adaptive, modality-invariant attribute kernels and bridge the modality gap while preserving attribute-aware semantic priors, each attribute is assigned $L$ learnable context tokens instead of static text embeddings. A context prompt is constructed for each attribute $a_k$:
\begin{equation}
P_k = [\mathbf{v}_k]_1 [\mathbf{v}_k]_2 \dots [\mathbf{v}_k]_L [a_k],
\end{equation}
where $[\mathbf{v_k}]_l$ represents learnable continuous context embeddings. These prompts are fed into the frozen CLIP text encoder $CLIP_T$ to generate $K$ modality-invariant semantic kernels:
\begin{equation}
\mathbf{c}_k = \text{Norm}(CLIP_T(P_k)) \in \mathbb{R}^{d \times 1}, \quad k \in \{1, \dots, K\}.
\end{equation}
The essence of attribute-prompted kernel generation lies in its capacity to utilize prompt tuning to forge an adaptive bridge between modalities without compromising the model's underlying semantic integrity. By optimizing learnable context tokens within the CAKM module, the generated kernels $\{\mathbf{c}_k\}$ evolve into dynamic semantic conduits that specifically adapt to the cross-modal manifold while strictly adhering to the high-level linguistic abstractions preserved from the pre-trained text branch. This mechanism is particularly pivotal in correspondence-constrained scenarios; by leveraging attribute-level universality, such as clustering diverse biological instances toward a common ``Furry" semantic center, the framework effectively compensates for the deficit of explicit instance-level alignment signals, ensuring robust generalization even with sparse training data. 

\textbf{Hyperspherical RBF Kernel Mapping.}
Unlike conventional CMH methods that rely on modality-specific hash functions, APKH reformulates cross-modal alignment as a structured feature mapping within a semantic kernel space. For any input feature $\mathbf{f}_i \in \{\mathbf{x}_i, \mathbf{y}_i\}$, its non-linear relationship to the $k$-th attribute kernel $\mathbf{c}_k$ is measured via a Gaussian Radial Basis Function (RBF) kernel:
\begin{equation}
\mathcal{K}_{RBF}(\mathbf{f}_i, \mathbf{c}_k) = \exp\left(-\frac{\|\mathbf{f}_i - \mathbf{c}_k\|^2}{2\sigma_k^2}\right),
\label{eq:rbf}
\end{equation}
where $\sigma_k$ controls the receptive field for the $k$-th attribute. To dynamically adapt the coverage of each attribute kernel based on its local semantic characteristics and accommodate the non-uniform density of the underlying CLIP manifold, we implement $\sigma_k$ as a learnable parameter optimized during training. This adaptive mechanism tailors the sensitivity of the $k$-th kernel, transforming it into a highly specialized, localized "domain expert" capable of capturing the distinct semantic nuances within its unique neighborhood.  

While standard RBF kernels operate on Euclidean distance, applying them to $\ell_2$-normalized representations inside a unit hypersphere $\mathbb{S}^{d-1}$ reveals a natural geometric congruence with modern contrastive image-language pre-training~\cite{radford2021learning}:
\begin{proposition}
For any $\ell_2$-normalized feature $\mathbf{f}_i$ and attribute kernel $\mathbf{c}_k$ (i.e., $\|\mathbf{f}_i\|_2 = \|\mathbf{c}_k\|_2 = 1$), the RBF kernel response $\mathcal{K}_{RBF}(\mathbf{f}_i, \mathbf{c}_k)$ is natively proportional to a temperature-scaled exponential cosine similarity used in contrastive learning~\cite{chen2020simple}.
\end{proposition}
\begin{IEEEproof}
Given the unit hyperspherical constraint $\|\mathbf{f}_i\|_2 = \|\mathbf{c}_k\|_2 = 1$, the squared Euclidean distance expands as:
\begin{equation}
\begin{aligned}
\|\mathbf{f}_i - \mathbf{c}_k\|^2 &= \|\mathbf{f}_i\|^2 + \|\mathbf{c}_k\|^2 - 2\mathbf{f}_i^T \mathbf{c}_k \\
&= 2 - 2\cos(\mathbf{f}_i, \mathbf{c}_k).
\end{aligned}
\end{equation}
Substituting this identity into Eq.~\eqref{eq:rbf} and defining the temperature parameter as $\tau_k = \sigma^2_k$ yields:
\begin{equation}
\begin{aligned}
\mathcal{K}_{RBF}(\mathbf{f}_i, \mathbf{c}_k) &= \exp\left(-\frac{2 - 2\cos(\mathbf{f}_i, \mathbf{c}_k)}{2\tau_k}\right) \\
&= \exp\left(-\frac{1}{\tau_k}\right) \cdot \exp\left(\frac{\cos(\mathbf{f}_i, \mathbf{c}_k)}{\tau_k}\right).
\end{aligned}
\end{equation}
Since $\exp(-1/\tau_k)$ is a constant with respect to the feature vector $\mathbf{f}_i$, we have:
\begin{equation}
\mathcal{K}_{RBF}(\mathbf{f}_i, \mathbf{c}_k) \propto \exp\left(\frac{\cos(\mathbf{f}_i, \mathbf{c}_k)}{\tau_k}\right).
\end{equation}
\end{IEEEproof}
The above proposition establishes a theoretical bridge between traditional RBF kernel methods and hyperspherical angular representation spaces, where the RBF kernel mapping of normalized features can be interpreted as a temperature-scaled exponential cosine similarity with any individual attribute kernel $\mathbf{c}_k$. By mapping features to kernel responses, APKH inherits the non-linear sensitivity of RBF networks without distorting CLIP's native angular topology, and provides a robust foundation for subsequent hash code generation. 

To transform the independent kernel responses into a unified semantic representation, we normalize them across all attribute anchors. This normalization converts the raw kernel activations into a relative distribution over the attribute vocabulary, where each response quantifies the semantic affinity between the input feature and a specific attribute kernel while inherently accounting for the receptive fields across all attributes:
\begin{equation}
S(\mathbf{f}_i, \mathbf{c}_k) = \frac{\mathcal{K}_{RBF}(\mathbf{f}_i, \mathbf{c}_k)}{\sum_{j=1}^K \mathcal{K}_{RBF}(\mathbf{f}_i, \mathbf{c}_j)}.
\end{equation}
Finally, the normalized responses across all $K$ kernels are concatenated to form a dense, modality-invariant kernel response vector $\mathbf{v}_i^* \in \mathbb{R}^K$:
\begin{equation}
\mathbf{v}_i^* = \big[ S(\mathbf{f}^*_i, \mathbf{c}_1), \dots, S(\mathbf{f}^*_i, \mathbf{c}_K) \big]^T, \;\; *\in\{x,y\}.
\end{equation}

\textbf{Unified Hash Function Encoding.}
By computing the hyperspherical RBF kernel responses across all $K$ attributes, the original heterogeneous image and text embeddings are implicitly mapped into a shared, modality-invariant semantic space. Specifically, the scalar responses are concatenated to form dense $K$-dimensional kernel response vectors, denoted as $\mathbf{v}_i^x \in \mathbb{R}^K$ and $\mathbf{v}_i^y \in \mathbb{R}^K$. Because these aggregated representations completely bridge the modality gap, APKH eliminates the need for complex, modality-specific hash functions. Instead, we feed both $\mathbf{v}_i^x$ and $\mathbf{v}_i^y$ into a single unified hash function, parameterized as a shared Multi-Layer Perceptron (MLP) denoted by $\Theta_H$. This unified network directly projects the $K$-dimensional kernel responses into a compact $B$-dimensional continuous Hamming space $\mathbf{h}_i^* = H(\mathbf{v}_i^*)$.
These relaxed hash codes $\mathbf{h}_i^x$ and $\mathbf{h}_i^y$ are then utilized as the input representations for the subsequent KSCA module to explicitly regularize the semantic manifold.

\subsection{Kernel-Smoothed Contrastive Alignment}
Standard cross-modal contrastive learning paradigms (e.g., InfoNCE loss) optimize the similarity between discrete point pairs. However, as established in the previous section, mapping sparse isolated points to a unified Hamming space inevitably results in fragile decision boundaries that fail to generalize to unseen categories. To mitigate this, we introduce Kernel-Smoothed Contrastive Alignment, which enables the hashing network to optimize over smoothed, continuous distributions. Instead of merely maximizing mutual information between individual image-text instances, our approach forces the network to align each anchor feature with the joint kernel distribution with its cross-modal counterpart.

\textbf{Kernel Distribution Generation.}
Specifically, KSCA transforms each discrete image-text pair into a localized, continuous probability distribution $\mathcal{P}_{i}$, which is modeled as a Gaussian kernel $\mathcal{N}(\boldsymbol{\mu}_{i}, \boldsymbol{\Sigma}_{i})$ within the CLIP latent space. This distribution is adaptively constructed by leveraging both the joint semantics and the empirical discrepancy of the modalities. The kernel centroid $\boldsymbol{\mu}_{i}$ is established as a semantic barycenter, computed by $\boldsymbol{\mu}_{i} = \text{Norm}(\mathbf{x}_i + \mathbf{y}_i)$. The kernel bandwidth is governed by an adaptive covariance $\boldsymbol{\Sigma}_{i} = \sigma_{i}^{2} \mathbf{I}$ to reflect the local semantic spread of the image-text pair around their shared semantic center, where $\sigma_{i}^{2}$ is the empirical variance of $\{\mathbf{x}_i, \mathbf{y}_i\}$ as $\sigma_{i}^{2} = \frac{1}{4} \|\mathbf{x}_i - \mathbf{y}_i\|_2^2$. The generated distribution $\mathcal{P}_{i}$ effectively populates the transitional semantic regions, allowing the model to capture the underlying manifold structure beyond sparse data points.

\textbf{Kernel-Smoothed Alignment Objective.}
To optimize the hashing network over these continuous distributions, we define the alignment objective as an expectation over the generated kernel distribution. Let $\mathbf{\tilde{f}} \sim \mathcal{P}_{i}$ denote the sample drawn from the variance-driven kernel distribution associated with the $i$-th pair and $\mathbf{\tilde{h}} = H(\Phi(\text{Norm}(\mathbf{\tilde{f}})))$ denote the corresponding virtual relaxed hash code.
The kernel-smoothed contrastive loss for the $i$-th pair is mathematically formulated as:
\begin{equation}
\mathcal{L}_{c}^{i} = - \sum\limits_{*} \log \frac{ \mathbb{E}_{\mathbf{\tilde{f}} \sim \mathcal{P}_{i}} \big[ \exp( (\mathbf{h}_{i}^{*})^T \mathbf{\tilde{h}} / \tau ) \big] }{ \mathcal{Z}^*_{i} },
\end{equation}
where $\tau$ is the temperature scaling parameter. Because the representations are $L_2$-normalized, the inner product $(\mathbf{h}_{i}^{*})^T \mathbf{\tilde{h}}$ directly measures their cosine similarity. The denominator $\mathcal{Z}^*_{i}$ represents the partition function across the entire batch, serving as the negative contrastive bounds:
\begin{equation}
\begin{split}
\mathcal{Z}^*_{i} ={} & \mathbb{E}_{\mathbf{\tilde{f}} \sim \mathcal{P}_{i}} \big[ \exp( (\mathbf{h}_{i}^{*})^T \mathbf{\tilde{h}} / \tau ) \big] \\
& + \sum_{j=1, j \neq i}^{N} \Big[ \exp( (\mathbf{h}_{i}^{*})^T \mathbf{h}_{j}^{*} / \tau ) + \exp( (\mathbf{h}_{i}^{*})^T \mathbf{h}_{j}^{\bar{*}} / \tau ) \Big],
\end{split}
\end{equation}
where $\bar{*}$ denotes the alternate modality of $*$. This formulation diverges fundamentally from conventional hashing losses. By optimizing the mathematical expectation directly, the network aligns the anchor with the entire continuous cross-modal manifold. 
In practice, this theoretical expectation is tractably approximated via Monte Carlo sampling during optimization, where $M$ virtual samples are randomly drawn from the kernel distribution $\mathcal{P}_{i}$ for each training epoch.

\subsection{Binarization and Overall Optimization}
\textbf{Binarization Error Reduction.}
Similar to prior methods~\cite{jiang2017deep,wang2017survey,tu2025cross,11341891,11084616}, to mitigate binarization errors caused by continuous relaxation and enhance the representational integrity of our hash functions, we incorporate a binarization loss function defined as:
\begin{equation}
    \mathcal{L}_{b}^{i}= \sum\limits_{*} \|\mathbf{h}_i^*-\text{Norm}(\mathbf{b}_i^*)\|^2_2.
	\label{eq:quant}
\end{equation}
This loss function ensures that $\mathbf{h}_i^*$ remains close to their corresponding binary hash codes $\mathbf{b}_i^*$. By guiding samples to approximate their binary counterparts, their respective pairwise neighborhoods are also drawn along, thereby facilitating a smoother transition to the Hamming space.

\textbf{Overall Optimization.}
By integrating $\mathcal{L}_{c}$ and $\mathcal{L}_{b}$, the overall optimization objective for the batch is defined as:
\begin{equation}
\mathcal{L} = \frac{1}{N} \sum_{i=1}^{N} \Big( \mathcal{L}_{c}^{i} + \alpha\mathcal{L}_{b}^{i} \Big),
\label{eq:overall}
\end{equation}
where $\alpha$ is a hyperparameter. The APKH optimization process can be summarized in Algorithm 1. The parameters of the unified hashing network $H(\cdot)$, along with the learnable context prompts in CAKM and the adaptive distribution parameters in KSCA, are optimized iteratively. For each iteration, the network and all learnable modules are updated through backpropagation by minimizing the overall kernel-smoothed contrastive loss $\mathcal{L}$ according to Eq.~\eqref{eq:overall}. This iterative optimization process continues until convergence, establishing a robust and modality-invariant Hamming space.

\begin{algorithm}[tb]
\caption{APKH Optimization}
\label{alg:apkh}
\textbf{Input}: Image-text features dataset $\mathcal{D}=\{\mathbf{x}_i, \mathbf{y}_i\}_{i=1}^n$, hash code length $B$, hyperparameter $\alpha$, and attributes $\{a_k\}_{k=1}^K$. \\
\textbf{Output}: Unified hashing network $H(\cdot)$ and learnable context prompts.

\begin{algorithmic}[1] 
\STATE Initialize learnable prompts $P_k$ for each attribute $a_k$.
\WHILE{not converged}
\FOR{each batch of image-text pairs}

\STATE Generate modality-invariant semantic kernels $\mathbf{c}_k$ using the frozen text encoder.
\STATE Compute Hyperspherical RBF responses and concatenate them to form kernel response vectors $\mathbf{v}_i^x, \mathbf{v}_i^y \in \mathbb{R}^K$.

\STATE Project the kernel responses into a continuous Hamming space via the unified MLP $H(\cdot)$ to obtain relaxed hash codes $\mathbf{h}_i^x, \mathbf{h}_i^y \in \mathbb{R}^B$.

\STATE Construct localized continuous probability distributions $\mathcal{P}_i$ over the cross-modal manifold to smooth the modality gap.
\STATE Compute the overall objective loss $\mathcal{L}$ (including the expectation over $\mathcal{P}_i$ and binarization error) according to Eq.~\eqref{eq:overall}.

\STATE Update the weights of the unified hashing network $H(\cdot)$ and the learnable prompts by minimizing $\mathcal{L}$ through backpropagation.

\ENDFOR
\ENDWHILE
\end{algorithmic}
\end{algorithm}

\section{Experiments}
\subsection{Implementation Details}
To evaluate our method, we utilize four widely-used benchmark datasets: MIR Flickr \cite{huiskes2008mir}, NUS-WIDE \cite{chua2009nus}, Pascal Sentence \cite{rashtchian2010collecting}, and Wikipedia \cite{rasiwasia2010new}. The specific configurations and splits for each dataset are detailed below:

\textbf{MIR Flickr (MIR)} comprises 25,000 images with associated text tags from Flickr. We filter this dataset by removing images with fewer than 20 tags, consistent with \cite{jiang2017deep}. This results in 20,015 valid image-text pairs. From this refined set, we randomly sample 2,000 pairs for the query set and use the remaining 18,015 for the retrieval set.

\textbf{NUS-WIDE (NUS)} contains 269,648 photos spanning 81 categories. Following \cite{Feng}, we restrict our focus to the top 10 most frequent categories to create a subset of 10,000 image-text pairs. We then randomly assign 2,000 pairs to the query set and the remaining 8,000 to the retrieval set.

\textbf{Pascal Sentence (PAS)} includes 1,000 image-text pairs distributed across 20 semantic classes, with each image described by five distinct sentences. Following \cite{zhen2019deep}, we partition the data into a query set of 200 pairs and a retrieval set of 800 pairs.

\textbf{Wikipedia (WIKI)} consists of 2,866 image-text pairs categorized into 10 semantic classes extracted from Wikipedia articles. Each image is described by a short, multi-sentence paragraph. In line with \cite{zhen2019deep}, we split this collection into a query set of 462 pairs and a retrieval set of 2,173 pairs.

For all datasets, the retrieval sets also serve as training sets, where training samples are randomly drawn from the retrieval sets with a fixed random seed of 1.
Following the convention \cite{hu2022UCCH,yu2021deep}, mean average precision (mAP) is adopted as the main criterion to evaluate the cross-modal retrieval performance. The mAP score is determined by computing the mean value of the average precision scores of all queries. To evaluate the model's retrieval performance and zero-shot generalization to unseen categories, we adopt a strict seen-to-unseen protocol. Specifically, the categories of each dataset are split equally into seen and unseen classes. During training, the model is exclusively optimized on samples whose labels strictly belong to the seen classes. During inference, both the query and retrieval gallery sets are constructed exclusively from either the seen or unseen classes.

We compare our method against seven state-of-the-art CMH methods, CIRH~\cite{zhu2022work}, CAGAN~\cite{li2023clip}, UCCH~\cite{hu2022UCCH}, CDTH~\cite{10463060}, DUHEG~\cite{pmlr-v267-song25h}, PIC-CMH~\cite{11341891} and GNAH~\cite{11084616}. Among these methods, CIRH, CAGAN, UCCH, and CDTH are unsupervised CMH methods. DUHEG is an unsupervised image hashing method that extracts external attribute features via clustering. We made minimal modifications to adapt it for cross-modal retrieval. PIC-CMH integrates learnable prompts directly into the modality-specific hash functions, it is uniquely provided with class labels during training, given its inherent supervised architecture. GNAH is our preliminary method of APKH.
To ensure a fair comparison, all methods use the frozen image and text encoders from CLIP ViT-B/16 as feature extractors. The APKH model is trained using the Adam~\cite{kingma2014adam} optimizer for 500 epochs. The learning rate adopted for training is 0.0001.
Following Visual Attributes in the Wild (VAW)~\cite{9578060}, we adopt a set of 620 unique attributes to represent universal semantic concepts (e.g., color, dimensions, materials, shapes, and textures), facilitating robust modality alignment.

\begin{table*}[t]
\renewcommand{\arraystretch}{1.1} 
\centering
\caption{Experimental results across different bit levels with 40 training pairs. The results are averaged over image-to-text and text-to-image retrieval. \textbf{Bold} denotes the best result, \underline{underline} denotes the second-best result.}
\label{tab:bit}
\footnotesize
\begin{tabular*}{0.99\textwidth}{@{\extracolsep{\fill}} cl|ccccc|ccccc|ccccc @{}}
\toprule
\multirow{2}{*}{\textbf{Bit}} & \multirow{2}{*}{\textbf{Method}} & \multicolumn{5}{c|}{\textbf{Seen}} & \multicolumn{5}{c|}{\textbf{Unseen}} & \multicolumn{5}{c}{\textbf{Average}} \\
\cmidrule{3-7} \cmidrule{8-12} \cmidrule{13-17}
& & MIR & NUS & PAS & WIKI & \textbf{AVG} & MIR & NUS & PAS & WIKI & \textbf{AVG} & MIR & NUS & PAS & WIKI & \textbf{AVG} \\
\midrule
\multirow{8}{*}{\shortstack{16 \\ bit}} 
& CIRH  & 0.854 & 0.362 & 0.304 & 0.352 & 0.468 & \textbf{0.717} & 0.224 & 0.243 & 0.275 & 0.365 & 0.786 & 0.293 & 0.273 & 0.313 & 0.416 \\
& CAGAN & 0.829 & \underline{0.485} & 0.397 & 0.248 & 0.490 & 0.711 & 0.269 & 0.278 & 0.234 & 0.373 & 0.770 & \underline{0.377} & 0.338 & 0.241 & 0.431 \\
& UCCH  & 0.841 & 0.323 & 0.371 & 0.331 & 0.466 & 0.714 & 0.238 & 0.220 & 0.302 & 0.368 & 0.778 & 0.280 & 0.296 & 0.316 & 0.417 \\
& CDTH  & 0.853 & 0.260 & 0.173 & 0.252 & 0.385 & 0.708 & 0.219 & 0.124 & 0.235 & 0.321 & 0.780 & 0.240 & 0.148 & 0.243 & 0.353 \\
& DUHEG & 0.823 & 0.280 & 0.302 & 0.247 & 0.413 & \textbf{0.717} & 0.221 & 0.178 & 0.257 & 0.343 & 0.770 & 0.250 & 0.240 & 0.252 & 0.378 \\
& PIC-CMH & 0.855 & 0.391 & 0.141 & \textbf{0.389} & 0.444 & 0.705 & 0.237 & 0.133 & 0.243 & 0.330 & 0.780 & 0.314 & 0.137 & 0.316 & 0.387 \\
\cmidrule{2-17} 
& GNAH  & \underline{0.860} & 0.452 & \underline{0.472} & 0.383 & \underline{0.542} & 0.708 & \underline{0.272} & \underline{0.280} & \underline{0.272} & \underline{0.383} & \underline{0.784} & 0.362 & \underline{0.376} & \underline{0.328} & \underline{0.462} \\
& APKH  & \textbf{0.869} & \textbf{0.503} & \textbf{0.509} & \underline{0.387} & \textbf{0.567} & 0.712 & \textbf{0.286} & \textbf{0.296} & \textbf{0.311} & \textbf{0.401} & \textbf{0.791} & \textbf{0.394} & \textbf{0.402} & \textbf{0.349} & \textbf{0.484} \\
\midrule
\multirow{7}{*}{\shortstack{32 \\ bit}} 
& CIRH  & 0.864 & 0.439 & 0.405 & 0.375 & 0.521 & 0.717 & 0.241 & \underline{0.271} & 0.275 & 0.376 & 0.790 & 0.340 & 0.338 & 0.325 & 0.448 \\
& CAGAN & 0.854 & \underline{0.472} & 0.430 & \underline{0.410} & 0.541 & \textbf{0.721} & 0.254 & 0.255 & 0.273 & 0.376 & 0.787 & 0.363 & 0.342 & 0.342 & 0.459 \\
& UCCH  & 0.845 & 0.375 & 0.379 & 0.332 & 0.483 & 0.709 & 0.235 & 0.212 & \underline{0.297} & 0.363 & 0.777 & 0.305 & 0.296 & 0.314 & 0.423 \\
& CDTH  & 0.840 & 0.258 & 0.178 & 0.258 & 0.383 & 0.706 & 0.210 & 0.132 & 0.235 & 0.321 & 0.773 & 0.234 & 0.155 & 0.246 & 0.352 \\
& DUHEG & 0.839 & 0.283 & 0.267 & 0.316 & 0.426 & 0.702 & 0.221 & 0.156 & 0.257 & 0.334 & 0.771 & 0.252 & 0.212 & 0.287 & 0.380 \\
& PIC-CMH  & 0.859 & 0.396 & 0.173 & 0.381 & 0.452 & 0.708 & 0.227 & 0.157 & 0.240 & 0.333 & 0.784 & 0.311 & 0.165 & 0.311 & 0.393 \\
\cmidrule{2-17} 
& GNAH  & \underline{0.870} & 0.465 & \underline{0.499} & \textbf{0.411} & \underline{0.561} & 0.716 & \underline{0.285} & 0.262 & 0.279 & \underline{0.385} & \underline{0.793} & \underline{0.375} & \underline{0.380} & \underline{0.345} & \underline{0.473} \\
& APKH  & \textbf{0.876} & \textbf{0.513} & \textbf{0.568} & 0.403 & \textbf{0.590} & \underline{0.719} & \textbf{0.298} & \textbf{0.292} & \textbf{0.351} & \textbf{0.415} & \textbf{0.798} & \textbf{0.405} & \textbf{0.430} & \textbf{0.377} & \textbf{0.503} \\
\midrule
\multirow{7}{*}{\shortstack{64 \\ bit}} 
& CIRH  & 0.855 & 0.432 & 0.465 & \underline{0.421} & 0.543 & 0.709 & 0.250 & \underline{0.307} & 0.285 & \underline{0.388} & 0.782 & 0.341 & \underline{0.386} & 0.353 & 0.466 \\
& CAGAN & 0.852 & 0.455 & 0.440 & 0.412 & 0.540 & 0.715 & 0.252 & 0.271 & 0.270 & 0.377 & 0.784 & 0.354 & 0.356 & 0.341 & 0.459 \\
& UCCH  & 0.840 & 0.415 & 0.457 & 0.361 & 0.518 & 0.704 & 0.260 & 0.270 & \underline{0.286} & 0.380 & 0.772 & 0.338 & 0.363 & 0.323 & 0.449 \\
& CDTH  & 0.854 & 0.298 & 0.269 & 0.294 & 0.429 & \underline{0.717} & 0.215 & 0.148 & 0.237 & 0.329 & 0.786 & 0.257 & 0.208 & 0.266 & 0.379 \\
& DUHEG & 0.842 & 0.346 & 0.322 & 0.291 & 0.450 & 0.709 & 0.222 & 0.181 & 0.250 & 0.340 & 0.776 & 0.284 & 0.251 & 0.270 & 0.395 \\
& PIC-CMH  & 0.863 & 0.419 & 0.203 & 0.399 & 0.471 & 0.708 & 0.236 & 0.156 & 0.245 & 0.336 & 0.786 & 0.328 & 0.180 & 0.322 & 0.404 \\
\cmidrule{2-17} 
& GNAH  & \underline{0.867} & \underline{0.518} & \underline{0.526} & \textbf{0.431} & \underline{0.585} & 0.714 & \underline{0.284} & 0.264 & 0.280 & 0.386 & \underline{0.791} & \underline{0.401} & \underline{0.395} & \underline{0.355} & \underline{0.486} \\
& APKH  & \textbf{0.874} & \textbf{0.525} & \textbf{0.592} & 0.405 & \textbf{0.599} & \textbf{0.719} & \textbf{0.301} & \textbf{0.335} & \textbf{0.346} & \textbf{0.425} & \textbf{0.797} & \textbf{0.413} & \textbf{0.463} & \textbf{0.376} & \textbf{0.512} \\
\midrule
\multirow{7}{*}{\shortstack{128 \\ bit}} 
& CIRH  & 0.861 & 0.482 & 0.476 & 0.403 & 0.555 & \underline{0.721} & 0.271 & \underline{0.315} & 0.278 & 0.396 & 0.791 & 0.376 & 0.395 & 0.340 & 0.476 \\
& CAGAN & 0.859 & 0.462 & 0.449 & 0.406 & 0.544 & 0.703 & 0.265 & 0.265 & 0.283 & 0.379 & 0.781 & 0.363 & 0.357 & 0.344 & 0.462 \\
& UCCH  & 0.845 & 0.438 & 0.457 & 0.371 & 0.528 & 0.712 & 0.252 & 0.243 & \underline{0.305} & 0.378 & 0.779 & 0.345 & 0.350 & 0.338 & 0.453 \\
& CDTH  & 0.850 & 0.383 & 0.414 & 0.344 & 0.498 & 0.705 & 0.221 & 0.164 & 0.249 & 0.335 & 0.777 & 0.302 & 0.289 & 0.296 & 0.416 \\
& DUHEG & 0.844 & 0.335 & 0.330 & 0.335 & 0.461 & 0.720 & 0.228 & 0.219 & 0.274 & 0.360 & 0.782 & 0.282 & 0.274 & 0.304 & 0.410 \\
& PIC-CMH & \textbf{0.876} & 0.500 & 0.219 & \underline{0.431} & 0.506 & 0.711 & 0.230 & 0.169 & 0.273 & 0.346 & \underline{0.793} & 0.365 & 0.194 & 0.352 & 0.426 \\
\cmidrule{2-17} 
& GNAH  & 0.864 & \textbf{0.557} & \underline{0.550} & 0.428 & \underline{0.600} & 0.712 & \underline{0.302} & 0.292 & 0.291 & \underline{0.399} & 0.788 & \underline{0.430} & \underline{0.421} & \underline{0.360} & \underline{0.500} \\
& APKH  & \textbf{0.876} & \underline{0.553} & \textbf{0.602} & \textbf{0.433} & \textbf{0.616} & \textbf{0.730} & \textbf{0.316} & \textbf{0.337} & \textbf{0.359} & \textbf{0.436} & \textbf{0.803} & \textbf{0.434} & \textbf{0.469} & \textbf{0.396} & \textbf{0.526} \\
\bottomrule
\end{tabular*}
\end{table*}

\subsection{Generalization from Seen to Unseen Semantics}
\textbf{Performance across Different Bit Lengths.}
To evaluate the impact of different hash code lengths on cross-modal retrieval performance, we present the mAP results across 16, 32, 64, and 128 bits under the heavily data-constrained scenario of 40 training pairs in Table~\ref{tab:bit}. The results demonstrate that our proposed APKH framework consistently and significantly outperforms all compared state-of-the-art unsupervised cross-modal hashing methods, including our preliminary method GNAH, across all four benchmark datasets and at every bit length. Even in the highly compact 16-bit configuration, APKH exhibits improved robustness, achieving an overall average mAP of 0.484, which notably surpasses the second-best method, GNAH (0.462), and other strong baselines like CAGAN (0.431). As the hash code length increases, the retrieval performance of the evaluated methods generally improves, reflecting the enhanced capacity of longer binary codes to encapsulate complex semantic information. At the length of 128 bits, APKH further widens the performance gap, yielding the highest overall average mAP of 0.526, compared to 0.500 for GNAH and 0.476 for CIRH.
Notably, the performance advantage of APKH over the baseline methods is significantly more pronounced on the unseen categories than on the seen ones. For instance, at 64 bits, APKH outperforms the second-best method (GNAH) by a margin of 0.014 in the seen average mAP (0.599 versus 0.585). However, in the unseen evaluation at the same bit length, APKH achieves a much larger lead of 0.037 over the top competing baseline (CIRH, which scores 0.388 compared to APKH's 0.425). This trend is consistent at 128 bits, where APKH's improvement margin over the second-best method on unseen data (0.037 over GNAH) is more than double its margin on seen data (0.016 over GNAH). This amplified performance gap in zero-shot scenarios underscores the core advantage of the proposed framework: rather than merely memorizing the limited training pairs, APKH effectively leverages modality-invariant semantic kernels and distribution smoothing to prevent overfitting, demonstrating improved generalization capabilities under severe data scarcity.

\textbf{Performance under Varying Degrees of Data Scarcity.}
To rigorously evaluate the correspondence efficiency of APKH, we analyze its retrieval performance by scaling the number of available training image-text pairs from 20 to 160 at a fixed code length of 64 bits. As illustrated in Fig.~\ref{fig:scarcity}, while all methods exhibit an expected upward trend in mAP as the supervision signals increase, APKH demonstrates a significantly higher performance floor and a more robust learning curve compared to existing state-of-the-art baselines. For the seen categories, APKH achieves competitive results even at the extreme 20 training pairs. While conventional methods like CAGAN and CIRH require a larger volume of correspondences to stabilize their modality-specific hash functions, APKH leverages the CAKM to effectively exploit the semantic priors of the frozen vision-language foundation model. This allows the model to establish a high-quality joint embedding space even when the instance-level alignment signals are sparse. The superiority of APKH is most evident in the unseen category evaluations, where it consistently maintains a substantial margin over the next-best competitors. Traditional unsupervised hashing frameworks often suffer from "manifold shattering" when training data is scarce, where the hash function overfits to the limited pairwise correlations and fails to preserve the underlying semantic structure for novel categories. In contrast, the KSCA module in APKH models discrete correspondences as continuous local distributions. 
Rather than merely memorizing isolated pairwise correlations, this distribution-driven alignment acts as a powerful statistical regularizer that preserves the intrinsic geometric integrity of the pre-trained vision-language manifold. By explicitly populating the transitional semantic regions, it prevents the discrete Hamming projection from overfitting to sparse correlations, thereby facilitating generalized retrieval performance.

\begin{figure}[]
    \centering
    \includegraphics[width=0.24\textwidth]{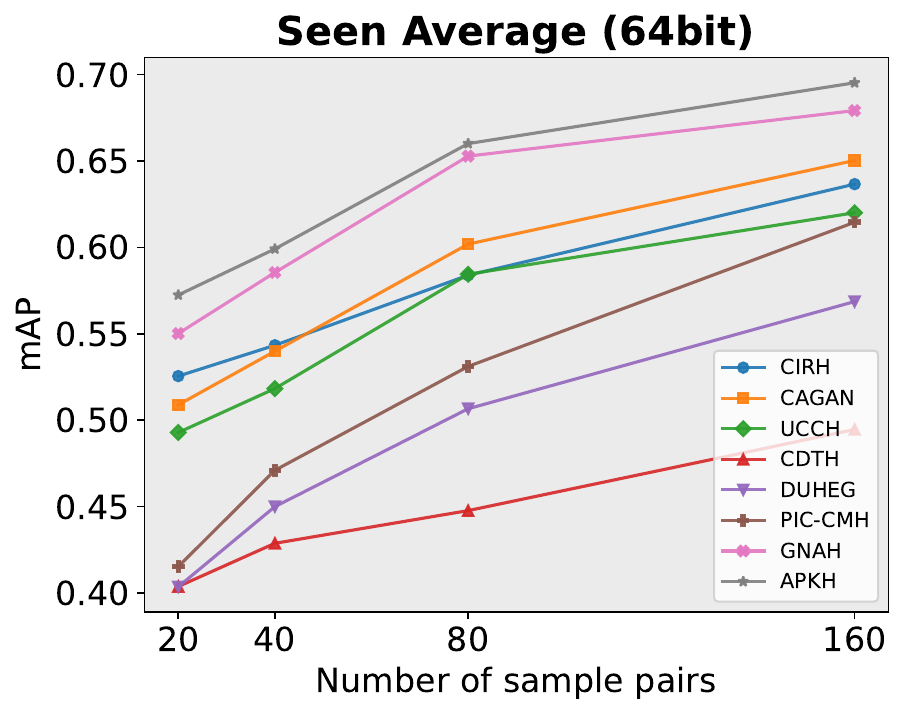}
    \includegraphics[width=0.24\textwidth]{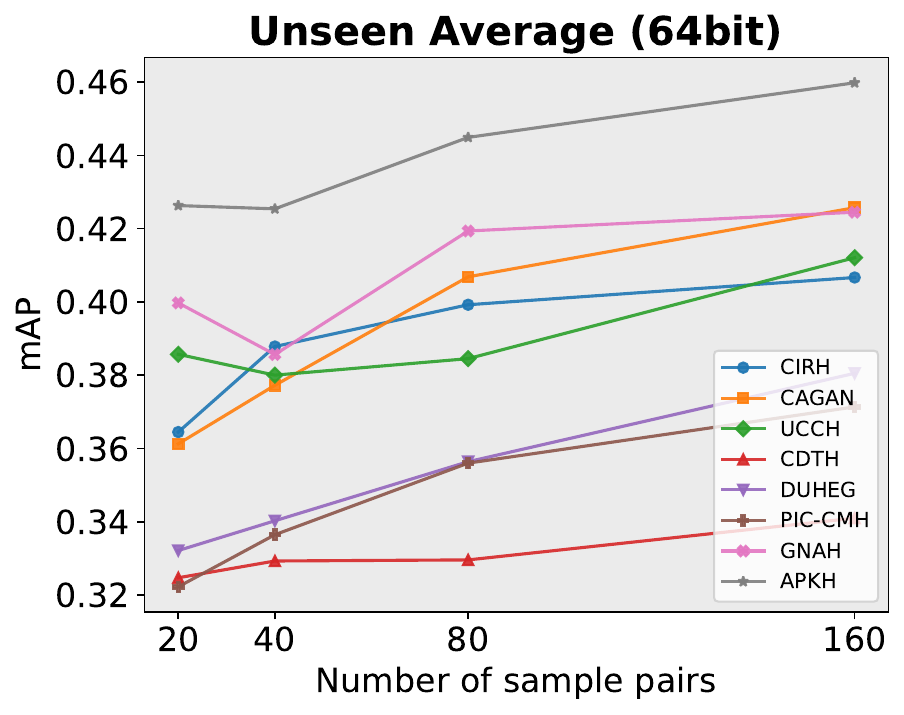}
    \includegraphics[width=0.24\textwidth]{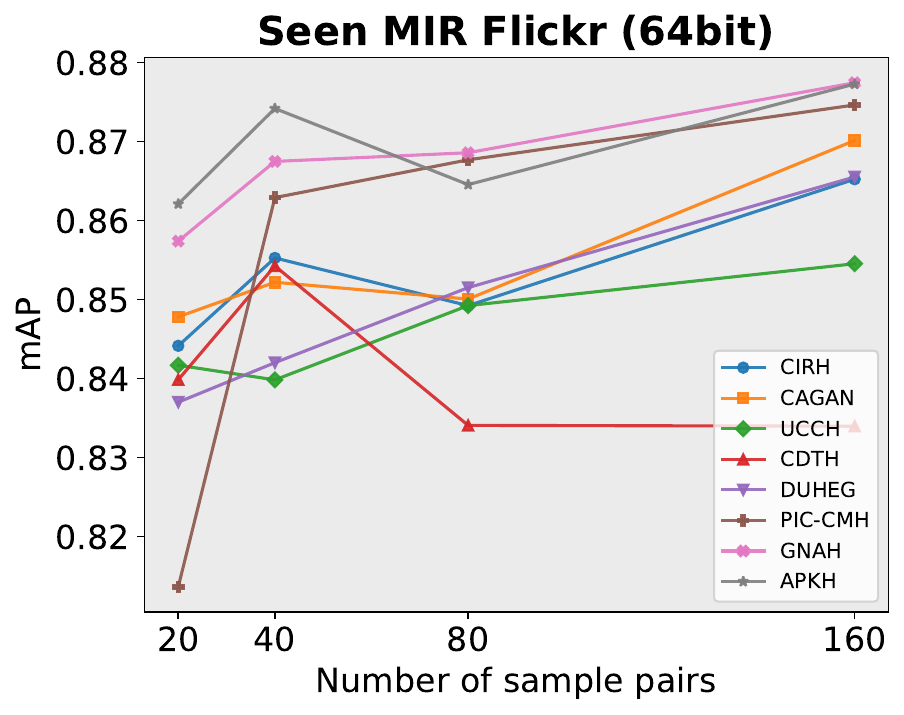}
    \includegraphics[width=0.24\textwidth]{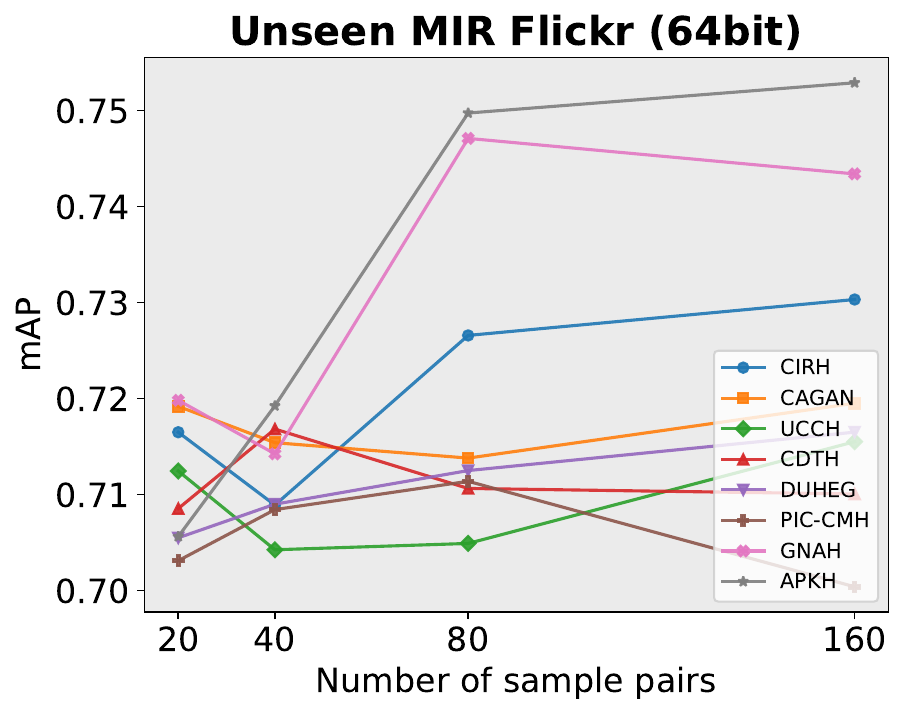}
    \includegraphics[width=0.24\textwidth]{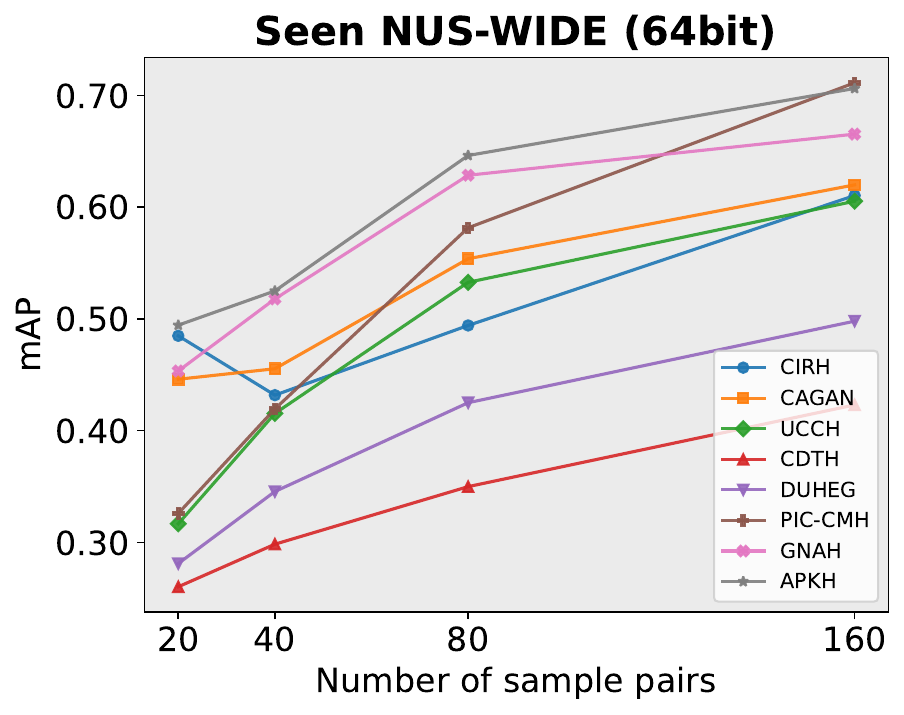}
    \includegraphics[width=0.24\textwidth]{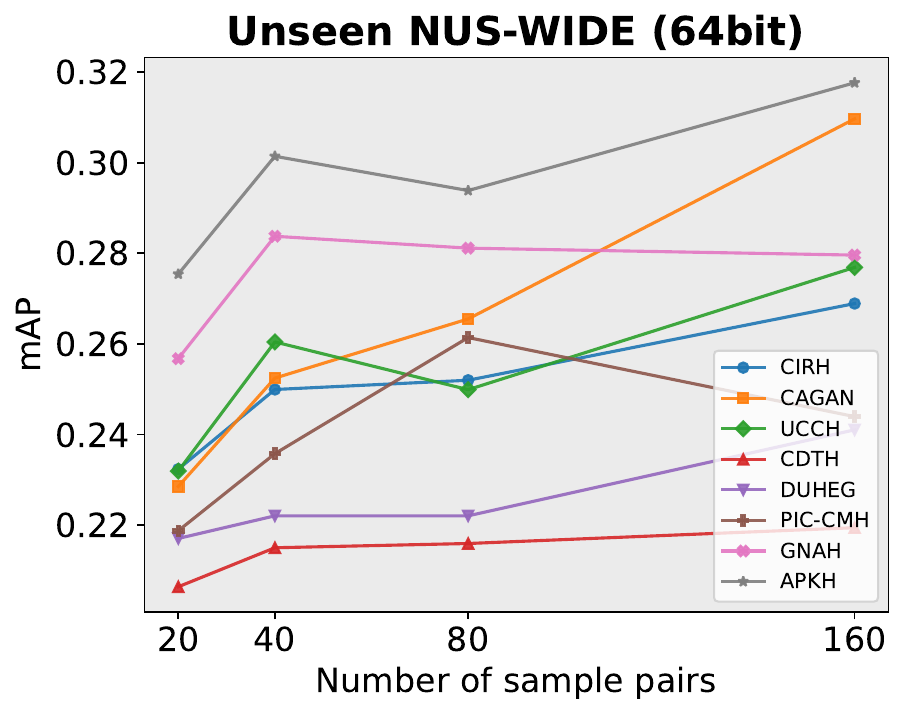}
    \includegraphics[width=0.24\textwidth]{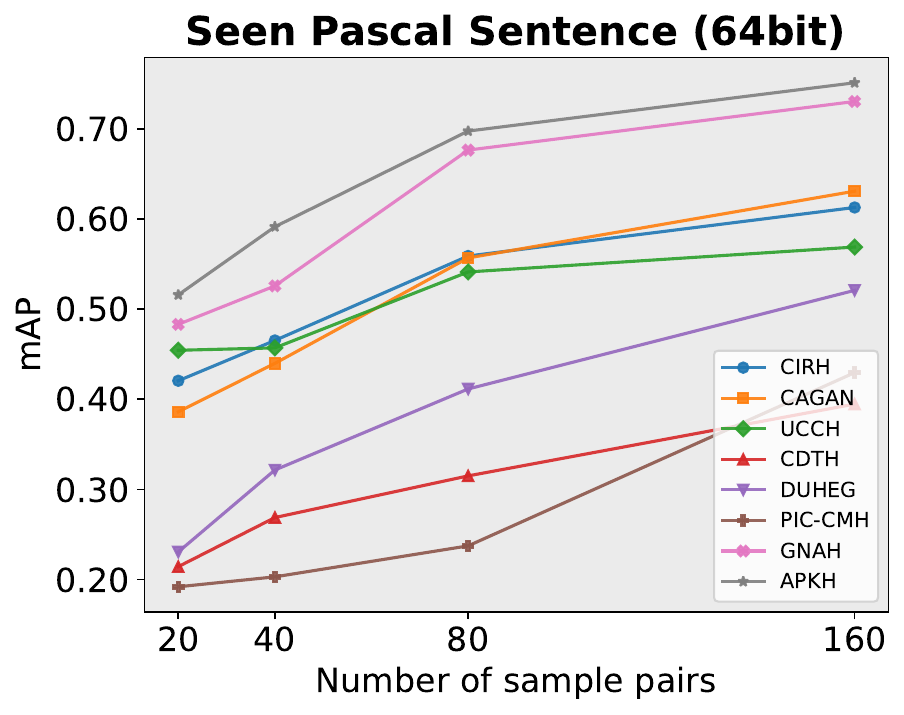}
    \includegraphics[width=0.24\textwidth]{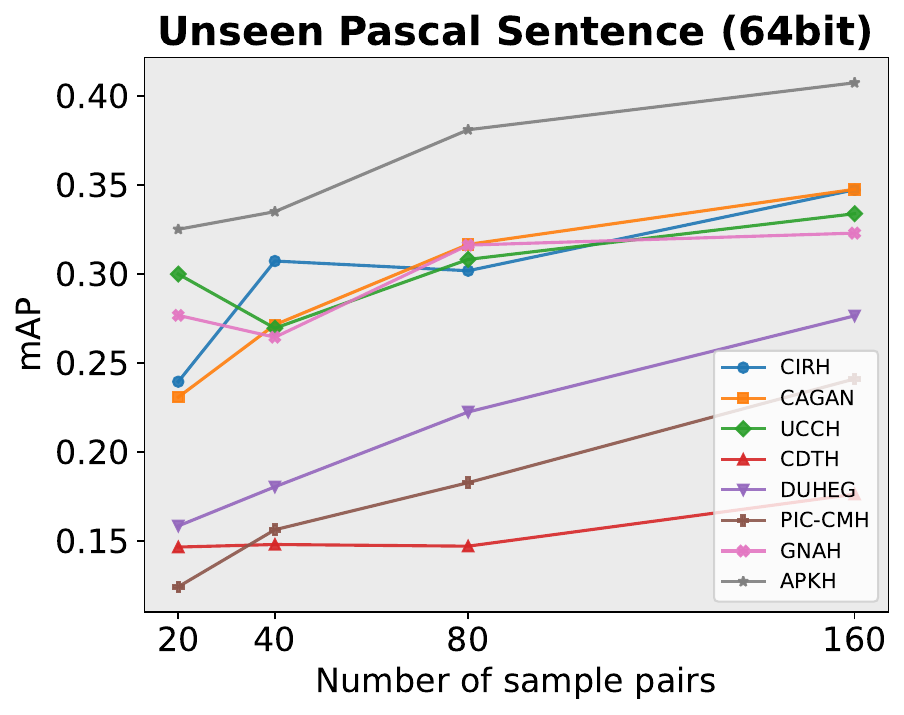}
    \includegraphics[width=0.24\textwidth]{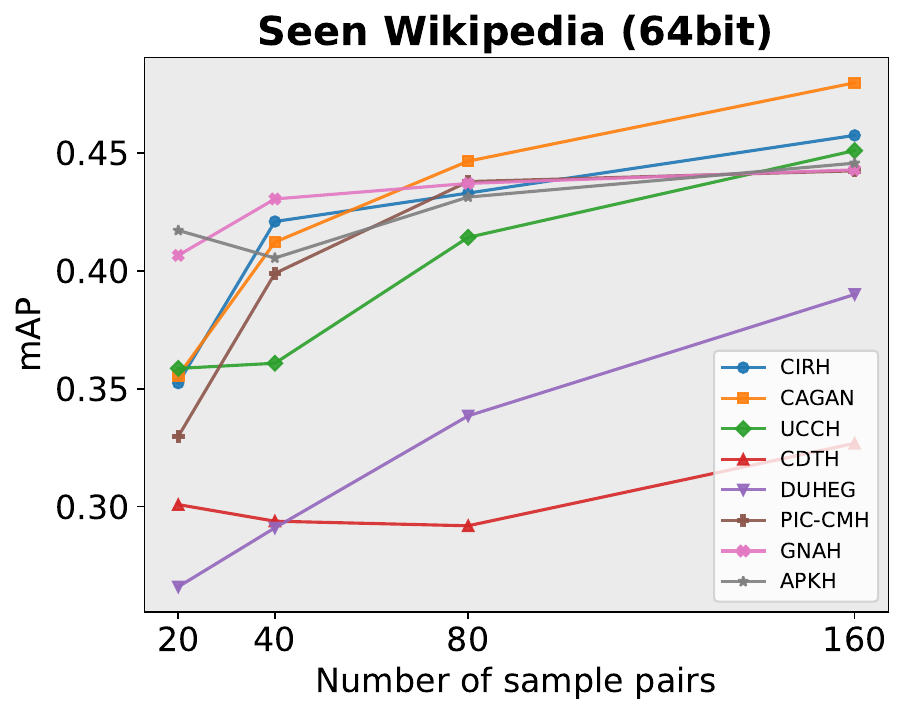}
    \includegraphics[width=0.24\textwidth]{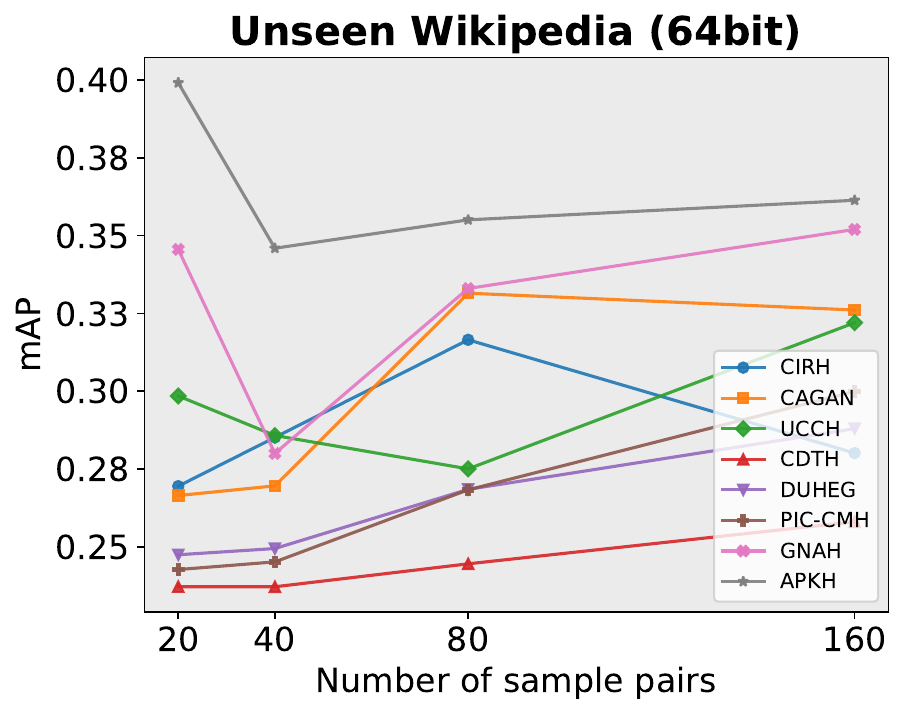} 
    \caption{Retrieval performance (mAP) of the proposed APKH and state-of-the-art baselines under varying degrees of data scarcity (20, 40, 80, and 160 training sample pairs) at 64 bits. The results are averaged over image-to-text and text-to-image retrieval.}
    \label{fig:scarcity}
\end{figure}

\begin{table}[]
  \centering
  \caption{Cross-dataset transfer retrieval performance (mAP) at 64 bits with 40 training pairs. \textbf{Bold} denotes the best result, \underline{underline} denotes the second-best result.}
  \label{tab:cross}
  \small
  \begin{tabular}{lccccc}
    \toprule
    \multirow{2}{*}{Method} & Source & \multicolumn{4}{c}{Target} \\
    \cmidrule(lr){2-2} \cmidrule(lr){3-6}
     & \textbf{MIR} & NUS & PAS & WIKI & \textbf{AVG} \\
    \midrule
    CIRH    & 0.855 & 0.285 & 0.162 & 0.256 & 0.235 \\
    CAGAN   & 0.852 & 0.285 & 0.130 & 0.249 & 0.221 \\
    UCCH    & 0.840 & 0.216 & 0.122 & 0.221 & 0.186 \\
    CDTH    & 0.854 & 0.239 & 0.119 & 0.233 & 0.197 \\
    DUHEG   & 0.842 & 0.240 & 0.153 & 0.241 & 0.211 \\
    PIC-CMH & 0.863 & 0.284 & 0.124 & 0.251 & 0.220 \\
    GNAH    & \underline{0.867} & \underline{0.376} & \textbf{0.201} & \underline{0.273} & \underline{0.283} \\
    APKH    & \textbf{0.874} & \textbf{0.406} & \underline{0.193} & \textbf{0.284} & \textbf{0.295} \\
    \midrule
     & \textbf{NUS} & MIR & PAS & WIKI & \textbf{AVG} \\
    \midrule
    CIRH    & 0.432 & 0.848 & 0.188 & 0.281 & 0.439 \\
    CAGAN   & 0.455 & 0.836 & 0.195 & 0.267 & 0.432 \\
    UCCH    & 0.415 & 0.838 & 0.185 & 0.284 & 0.436 \\
    CDTH    & 0.298 & 0.830 & 0.127 & 0.229 & 0.395 \\
    DUHEG   & 0.346 & 0.848 & 0.147 & 0.266 & 0.420 \\
    PIC-CMH & 0.419 & 0.821 & 0.129 & 0.289 & 0.413 \\
    GNAH    & \underline{0.518} & \textbf{0.866} & \underline{0.269} & \underline{0.296} & \underline{0.477} \\
    APKH    & \textbf{0.525} & \underline{0.863} & \textbf{0.271} & \textbf{0.334} & \textbf{0.489} \\
    \midrule
     & \textbf{PAS} & MIR & NUS & WIKI & \textbf{AVG} \\
    \midrule
    CIRH    & 0.465 & 0.858 & 0.247 & 0.242 & 0.449 \\
    CAGAN   & 0.440 & 0.843 & 0.248 & \underline{0.287} & 0.459 \\
    UCCH    & 0.457 & 0.806 & 0.237 & 0.259 & 0.434 \\
    CDTH    & 0.269 & 0.820 & 0.217 & 0.230 & 0.422 \\
    DUHEG   & 0.322 & 0.841 & 0.237 & 0.247 & 0.441 \\
    PIC-CMH & 0.203 & 0.825 & 0.217 & 0.251 & 0.431 \\
    GNAH    & \underline{0.526} & \underline{0.861} & \underline{0.271} & 0.286 & \underline{0.473} \\
    APKH    & \textbf{0.592} & \textbf{0.870} & \textbf{0.343} & \textbf{0.297} & \textbf{0.503} \\
    \midrule
     & \textbf{WIKI} & MIR & NUS & PAS & \textbf{AVG} \\
    \midrule
    CIRH    & \underline{0.421} & 0.832 & 0.266 & 0.173 & 0.424 \\
    CAGAN   & 0.412 & 0.835 & 0.240 & 0.163 & 0.413 \\
    UCCH    & 0.361 & 0.832 & 0.256 & 0.146 & 0.411 \\
    CDTH    & 0.294 & 0.830 & 0.224 & 0.140 & 0.398 \\
    DUHEG   & 0.291 & 0.834 & 0.243 & 0.157 & 0.411 \\
    PIC-CMH & 0.399 & 0.818 & \underline{0.273} & 0.135 & 0.409 \\
    GNAH    & \textbf{0.431} & \underline{0.842} & 0.256 & \underline{0.201} & \underline{0.433} \\
    APKH    & 0.405 & \textbf{0.847} & \textbf{0.318} & \textbf{0.224} & \textbf{0.463} \\
    \bottomrule
  \end{tabular}
\end{table}

\textbf{Precision-Recall Curve.}
To further evaluate retrieval quality, Fig.~\ref{fig:prResults} illustrates the Precision-Recall (PR) curves for both Image-to-Text (I2T) and Text-to-Image (T2I) tasks on the NUS-WIDE and Wikipedia datasets at 64 bits with 40 training pairs. On the NUS-WIDE dataset and across the unseen categories of both datasets, the proposed APKH consistently encapsulates the curves of competing baselines, achieving a significantly larger Area Under the Curve (AUC). This demonstrates that APKH maintains higher precision at equivalent recall levels, delivering more accurate top-ranked matches. It is worth noting that on the seen categories of Wikipedia, the strict statistical regularization in APKH may slightly constrain its peak fitting performance, making its curves slightly lower than some baselines. Crucially, however, the performance advantage of APKH becomes most pronounced in the unseen zero-shot scenarios. While baseline methods experience rapid precision degradation as recall increases due to overfitted, shattered manifolds, APKH exhibits a much smoother and more stable decline.

\begin{figure*}[]
    \centering
    \includegraphics[width=0.24\textwidth]{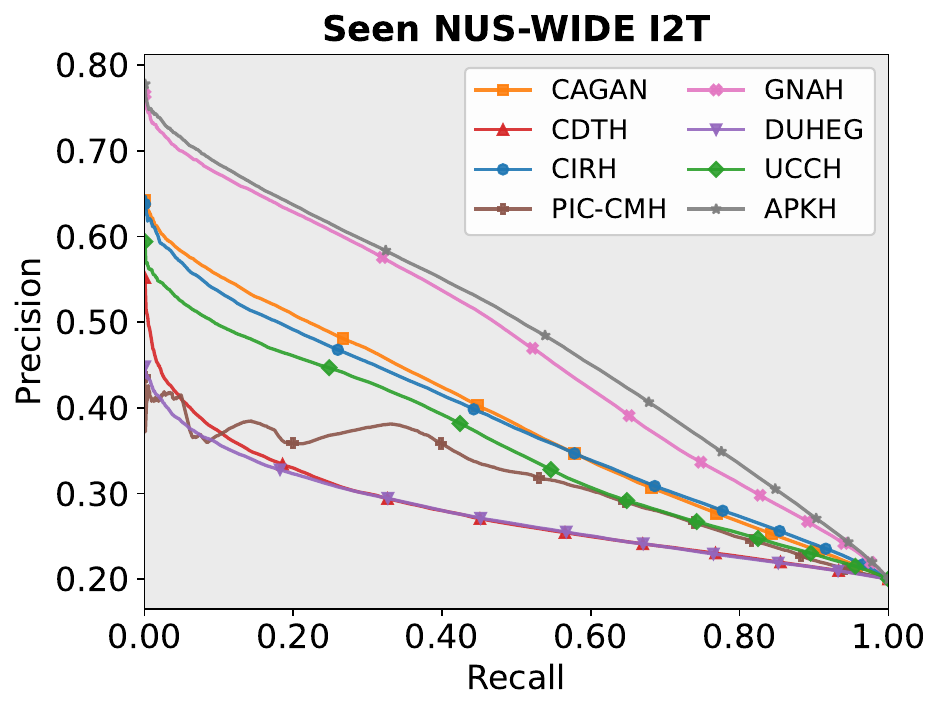}
    \includegraphics[width=0.24\textwidth]{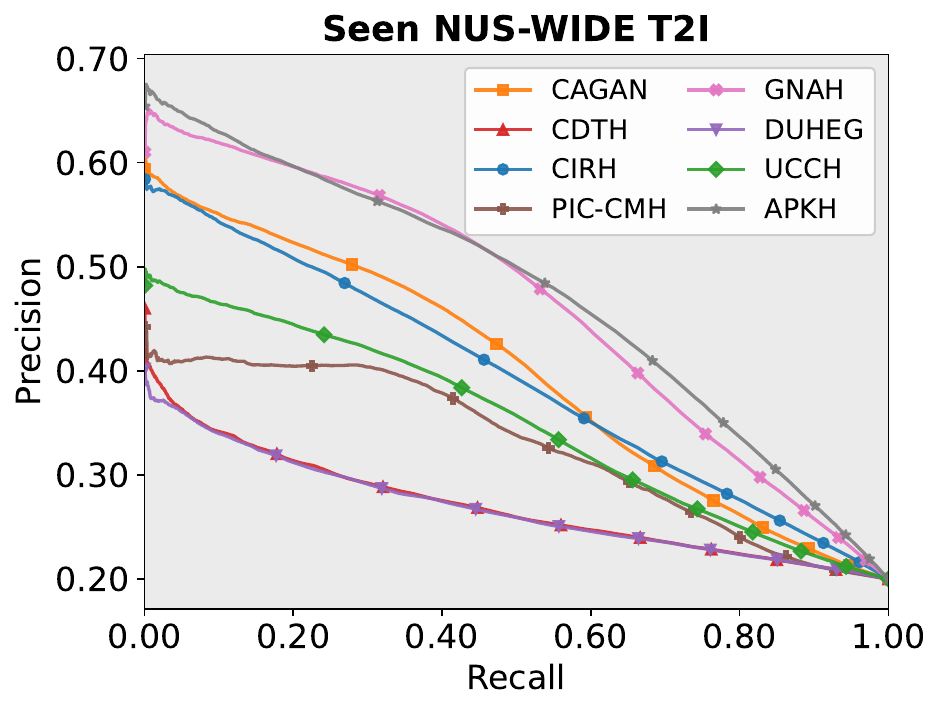}
    \includegraphics[width=0.24\textwidth]{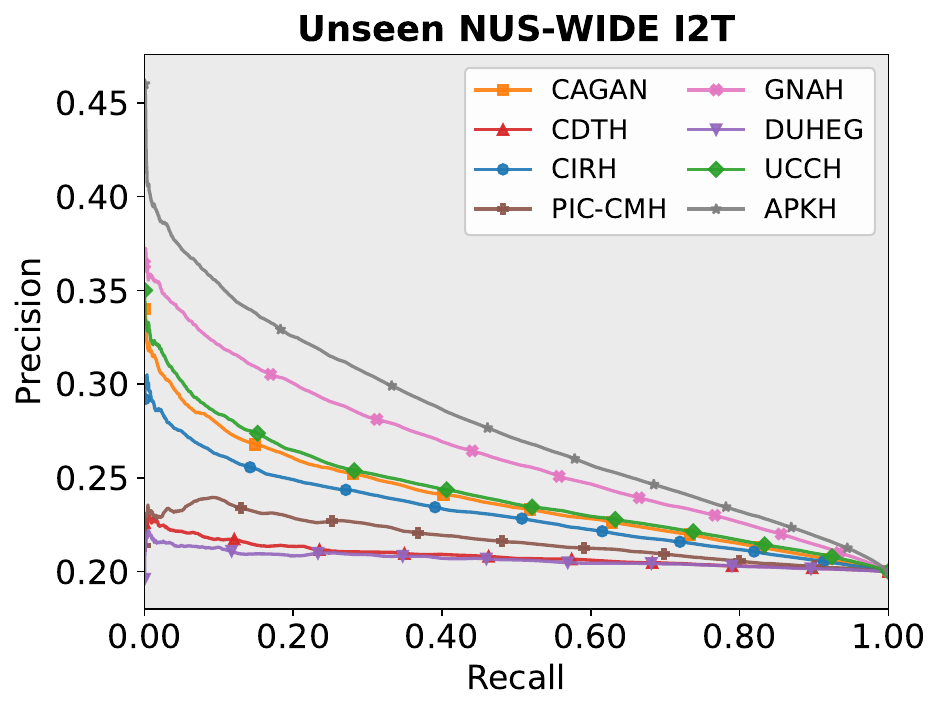}
    \includegraphics[width=0.24\textwidth]{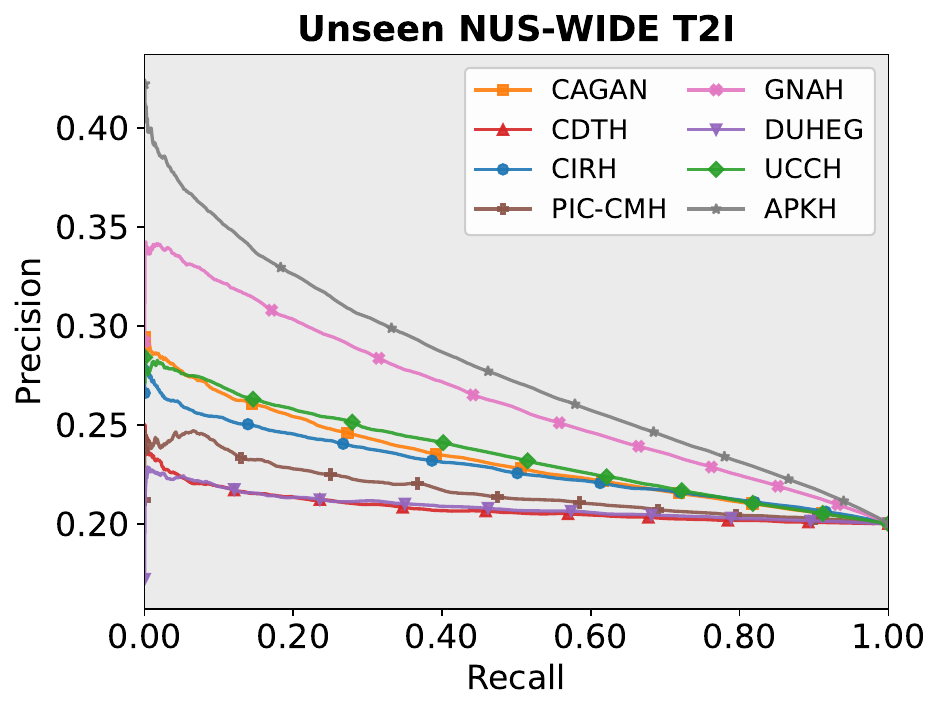}
    \includegraphics[width=0.24\textwidth]{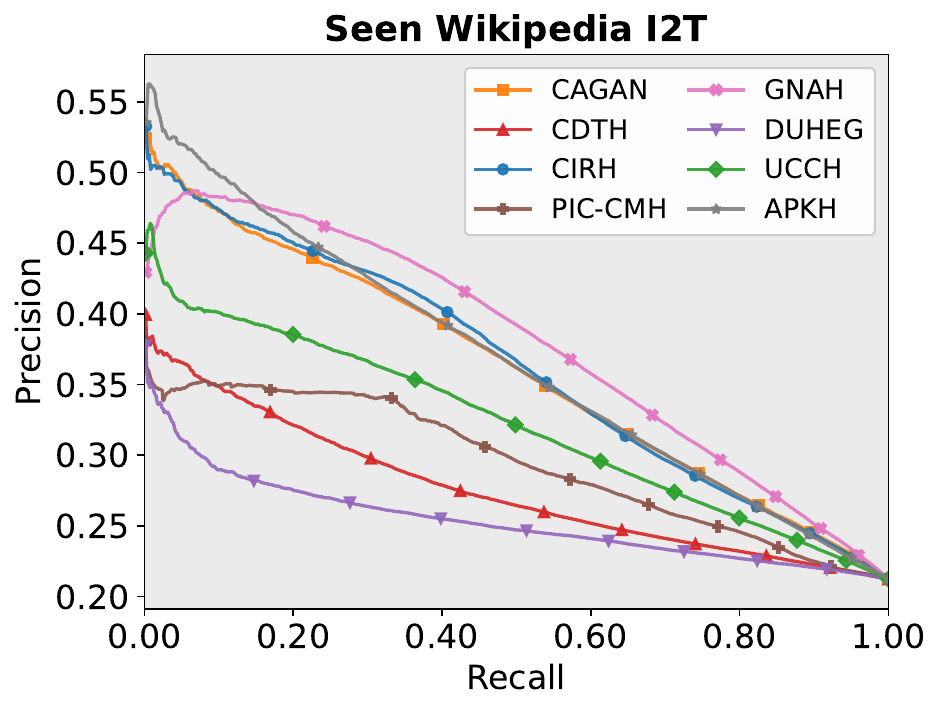}
    \includegraphics[width=0.24\textwidth]{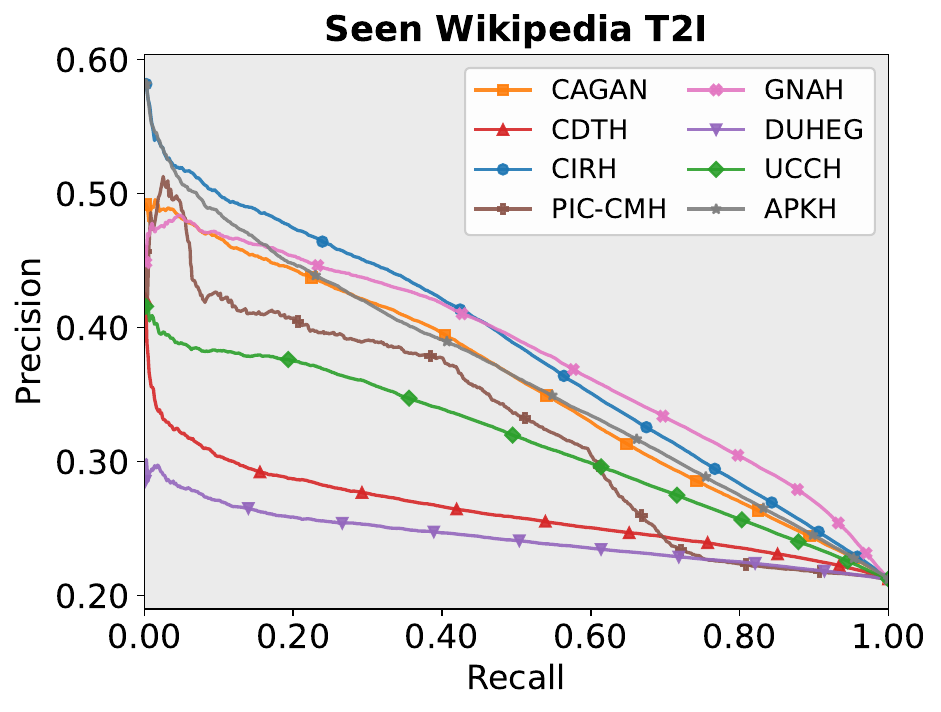}
    \includegraphics[width=0.24\textwidth]{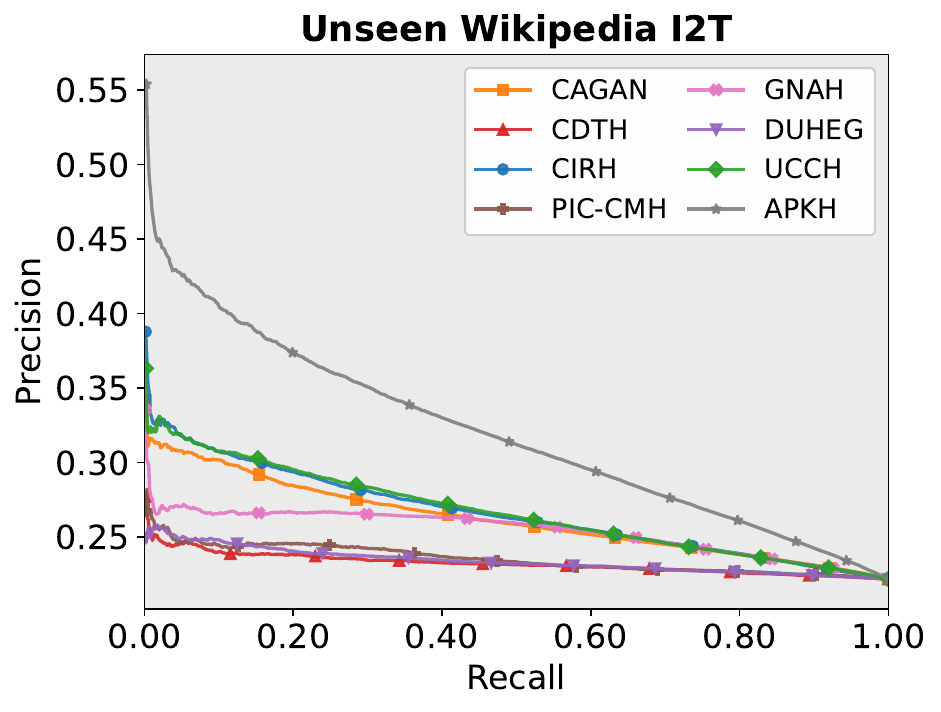}
    \includegraphics[width=0.24\textwidth]{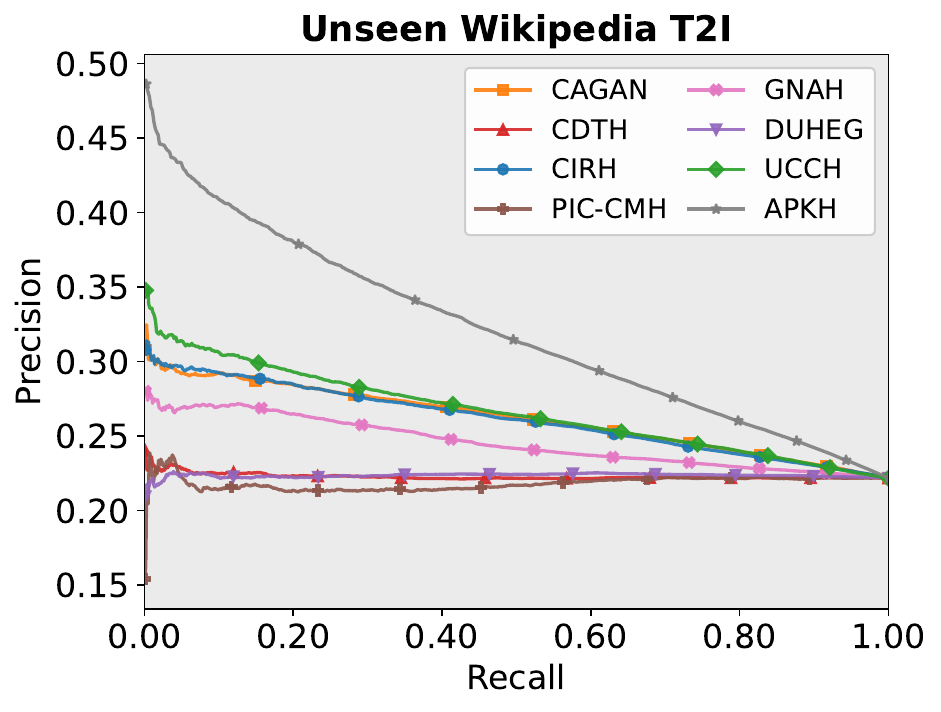}

    \caption{Precision-Recall (PR) curves for image-to-text (I2T) and text-to-image (T2I) retrieval tasks on NUS-WIDE and Wikipedia datasets. The results compare our proposed APKH with several state-of-the-art methods using 64-bit hash codes and 40 training pairs across both seen and unseen categories.}
    \label{fig:prResults}
\end{figure*}

\begin{figure*}[]
    \centering
    \includegraphics[width=0.24\textwidth]{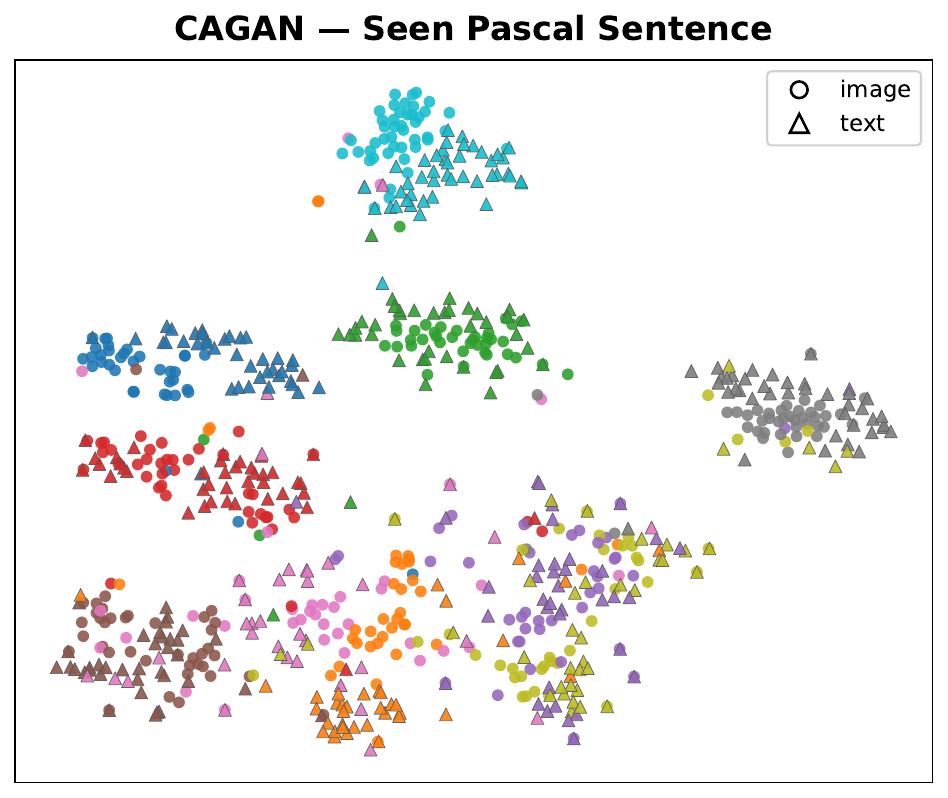}
    \includegraphics[width=0.24\textwidth]{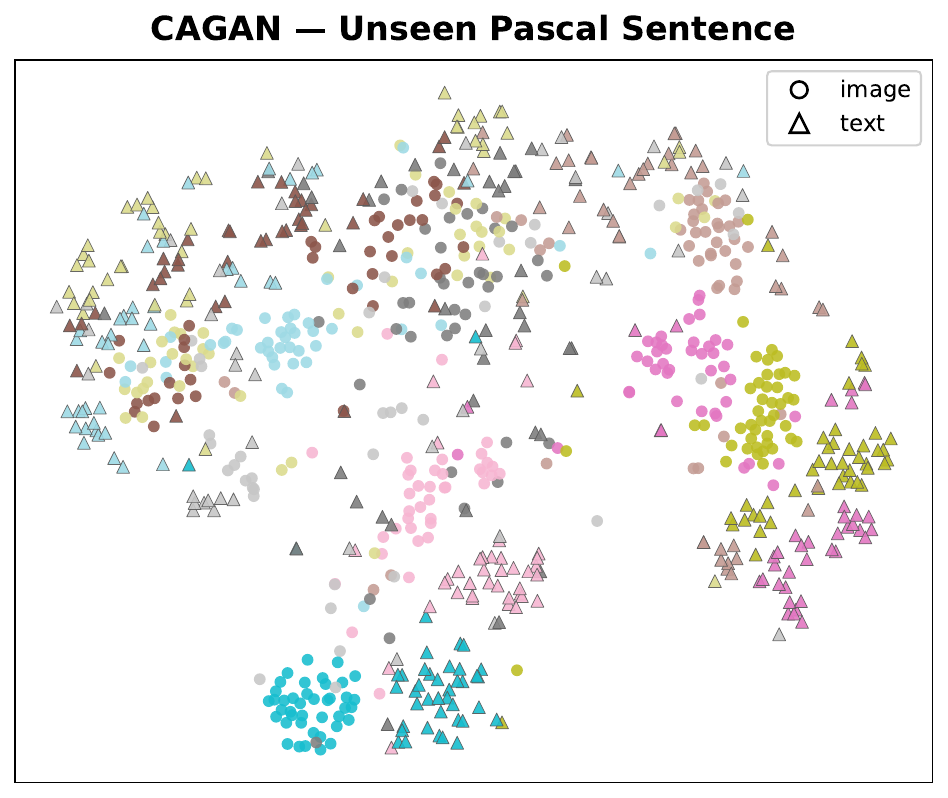}
    \includegraphics[width=0.24\textwidth]{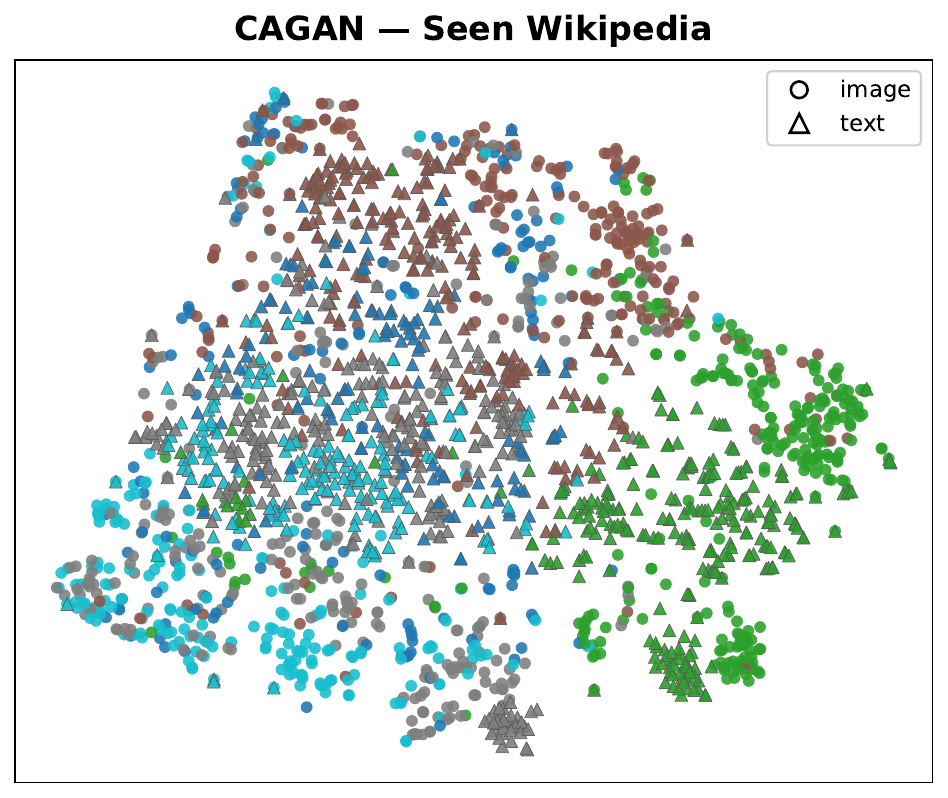}
    \includegraphics[width=0.24\textwidth]{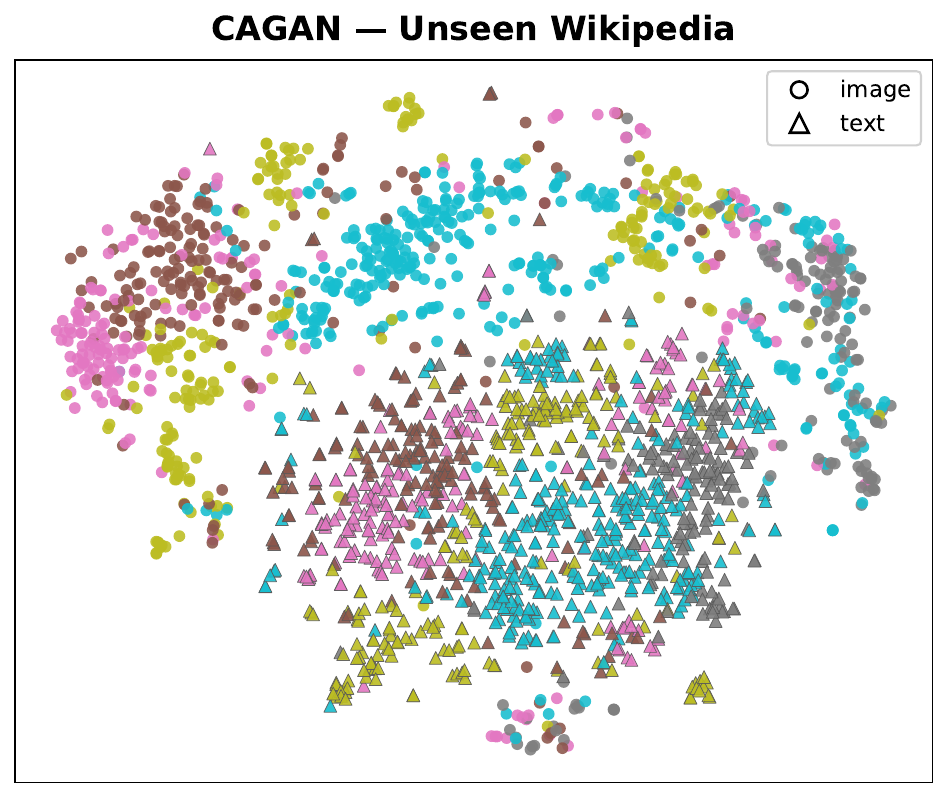}
    \includegraphics[width=0.24\textwidth]{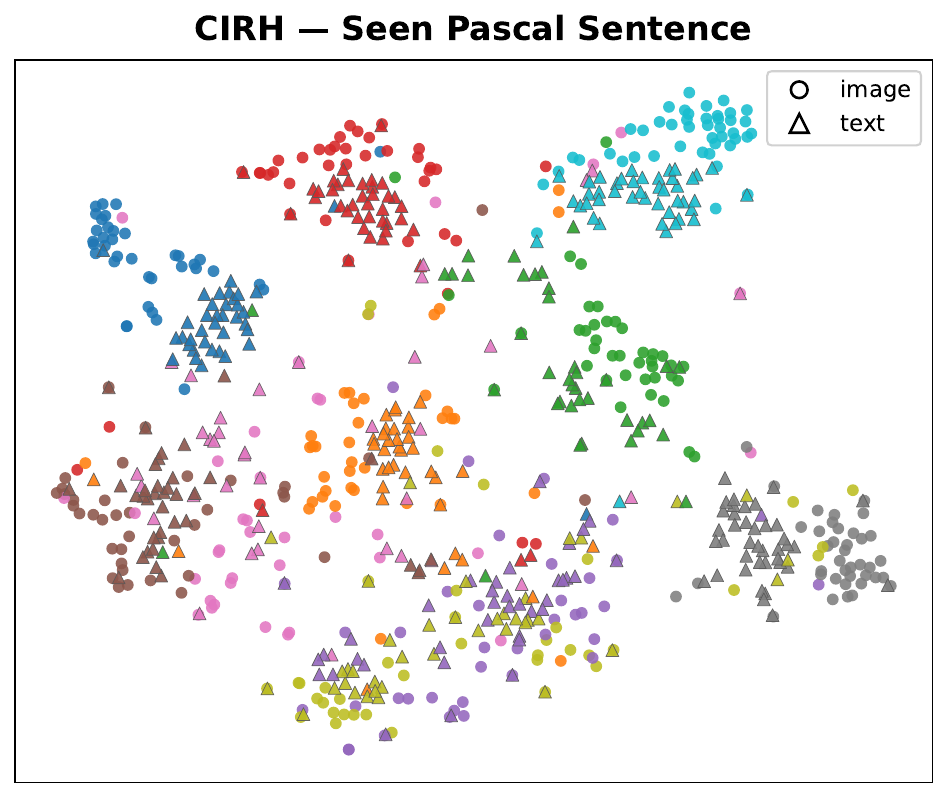}
    \includegraphics[width=0.24\textwidth]{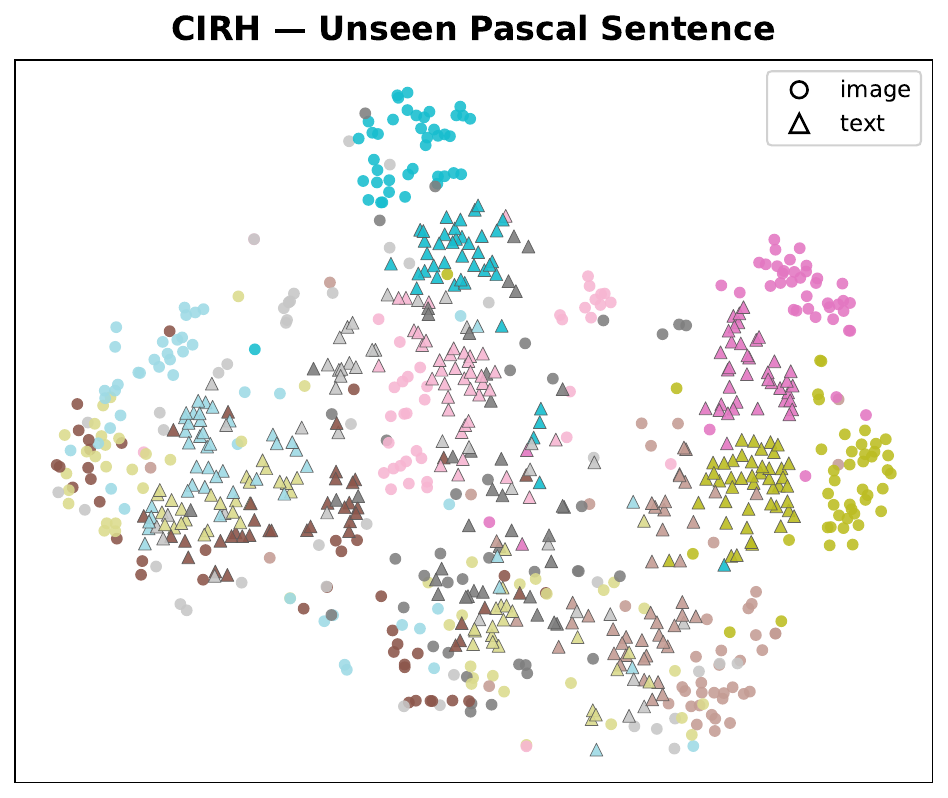}
    \includegraphics[width=0.24\textwidth]{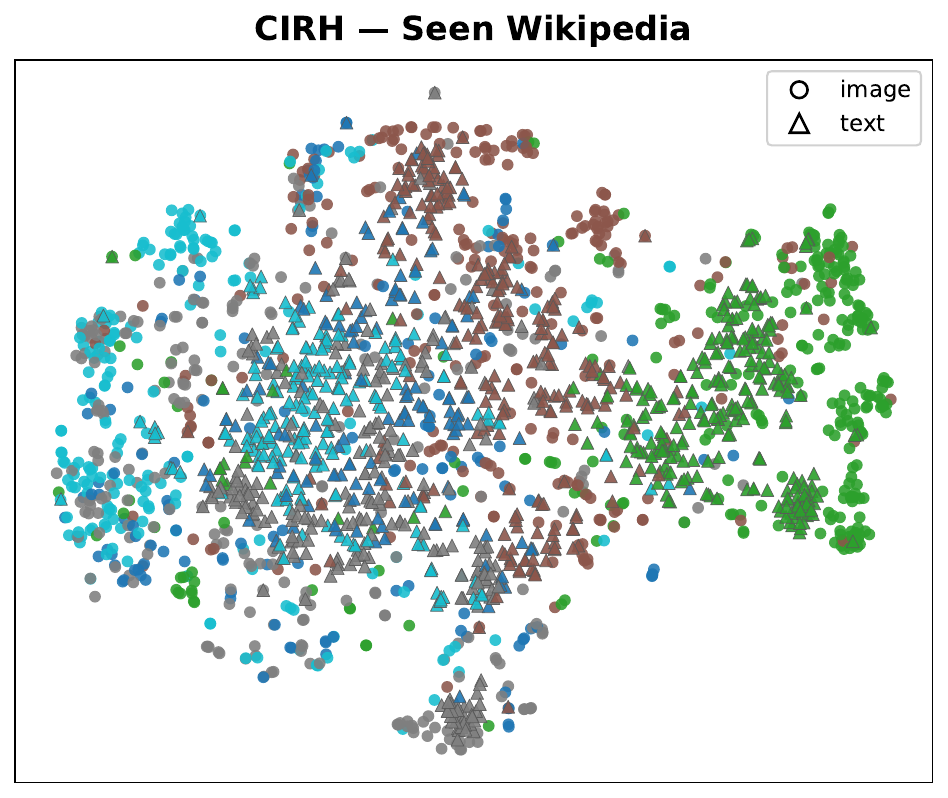}
    \includegraphics[width=0.24\textwidth]{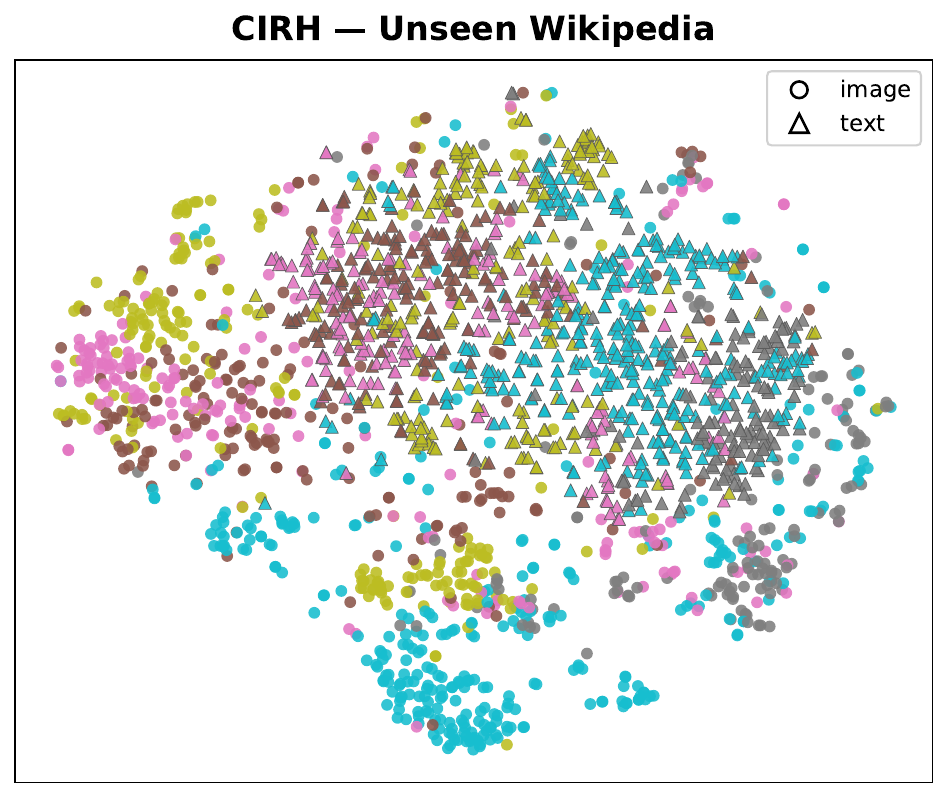}
    \includegraphics[width=0.24\textwidth]{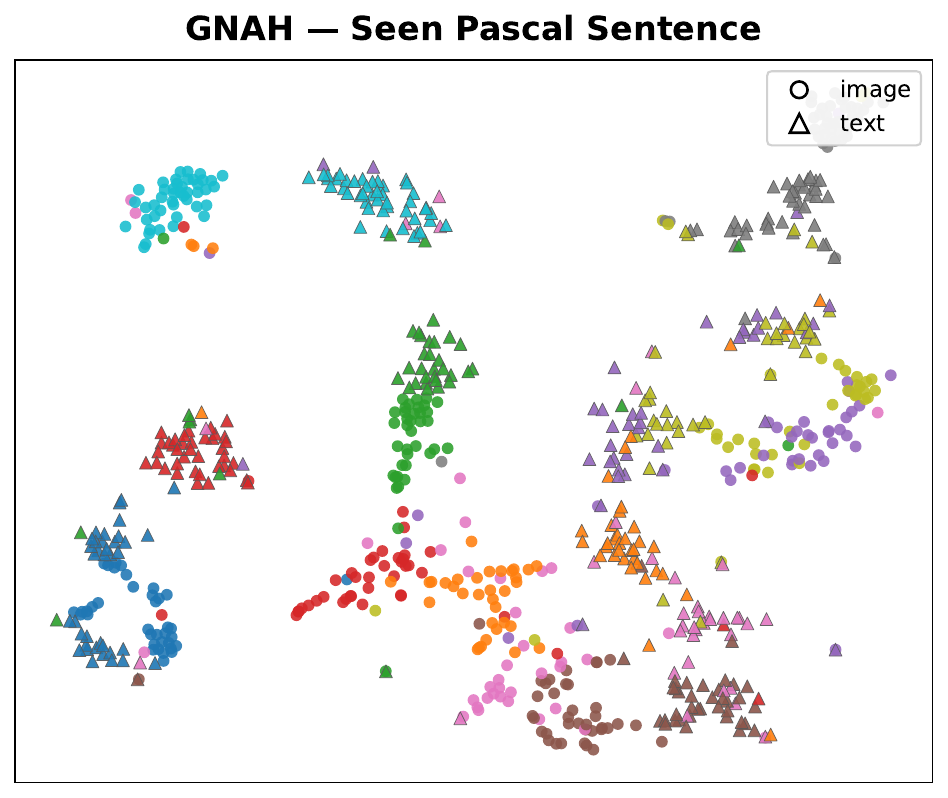}
    \includegraphics[width=0.24\textwidth]{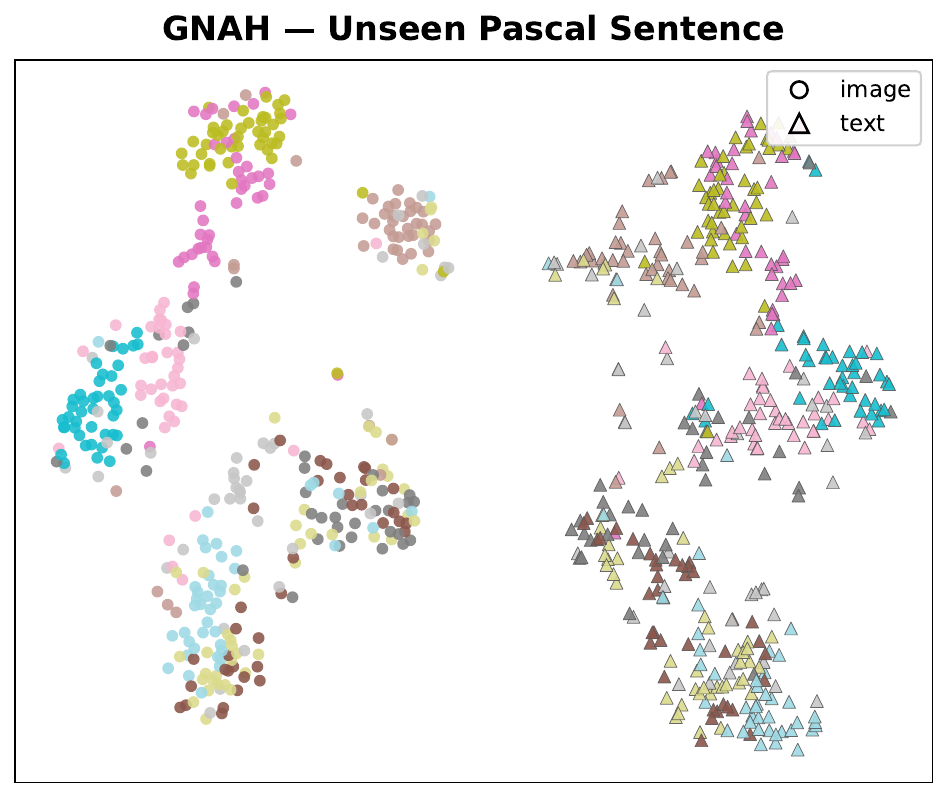}
    \includegraphics[width=0.24\textwidth]{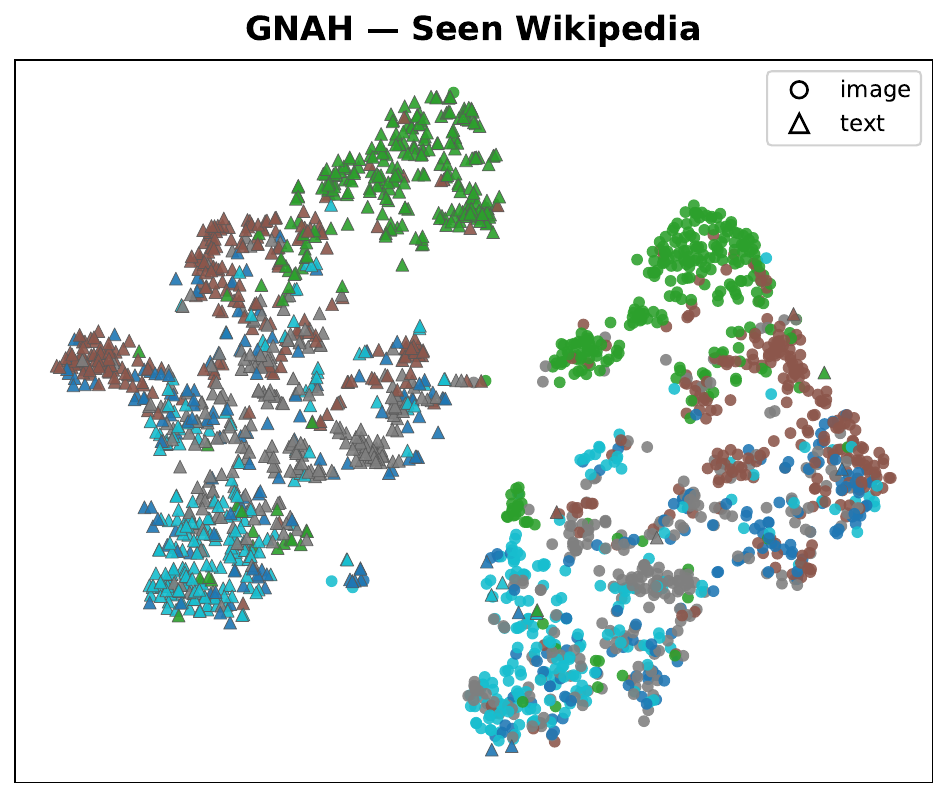}
    \includegraphics[width=0.24\textwidth]{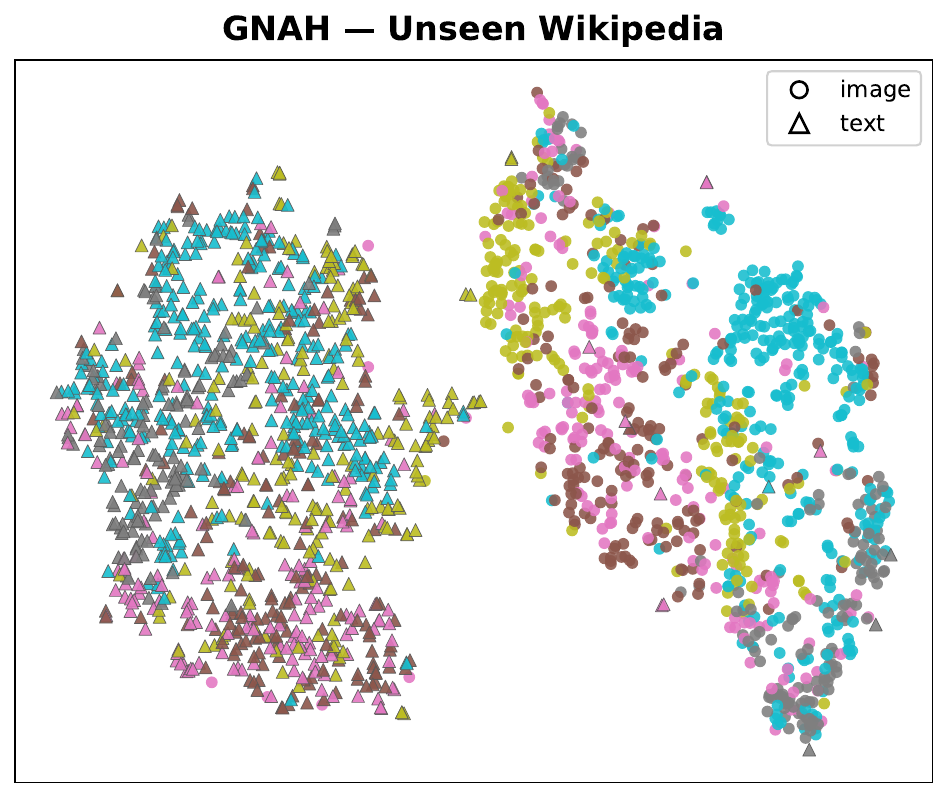}
    \includegraphics[width=0.24\textwidth]{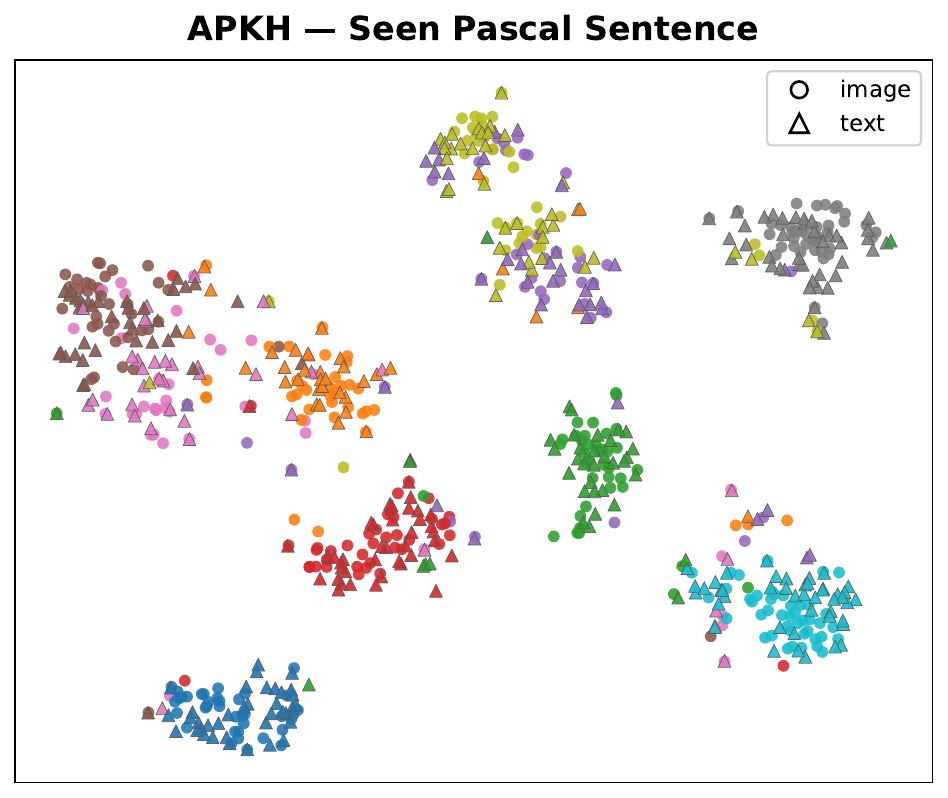}
    \includegraphics[width=0.24\textwidth]{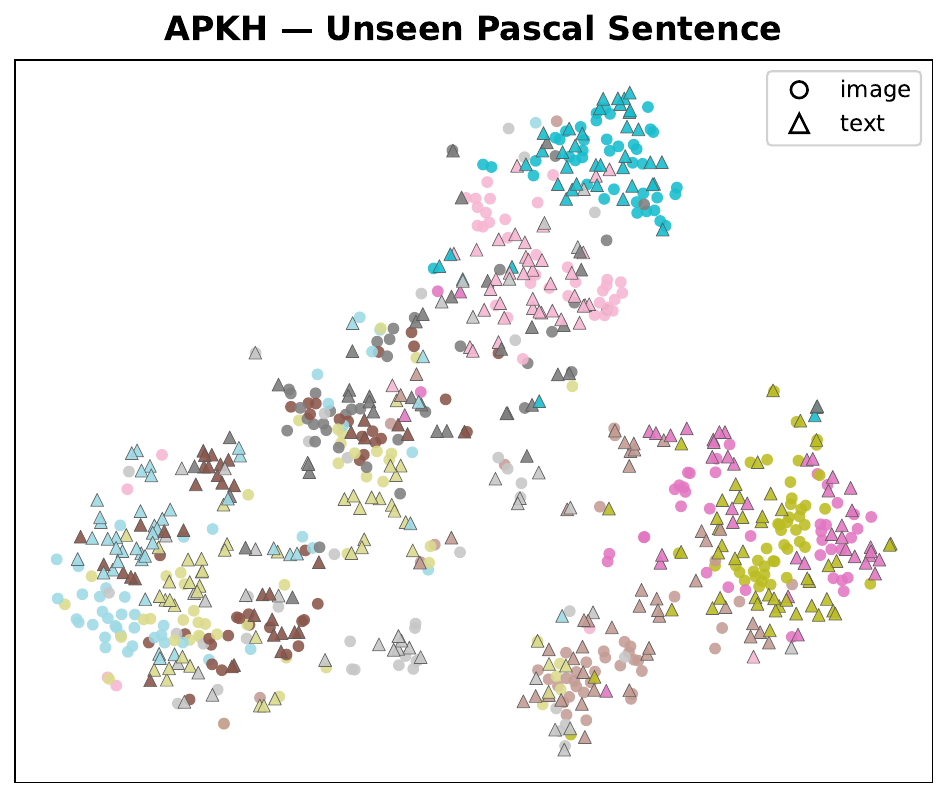}
    \includegraphics[width=0.24\textwidth]{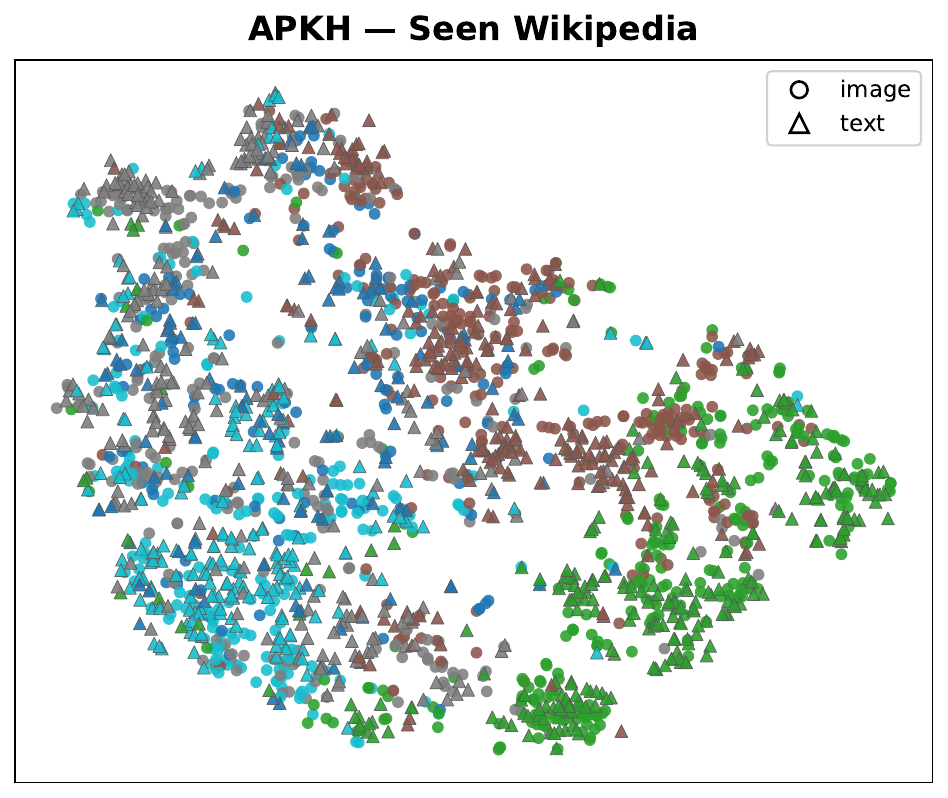}
    \includegraphics[width=0.24\textwidth]{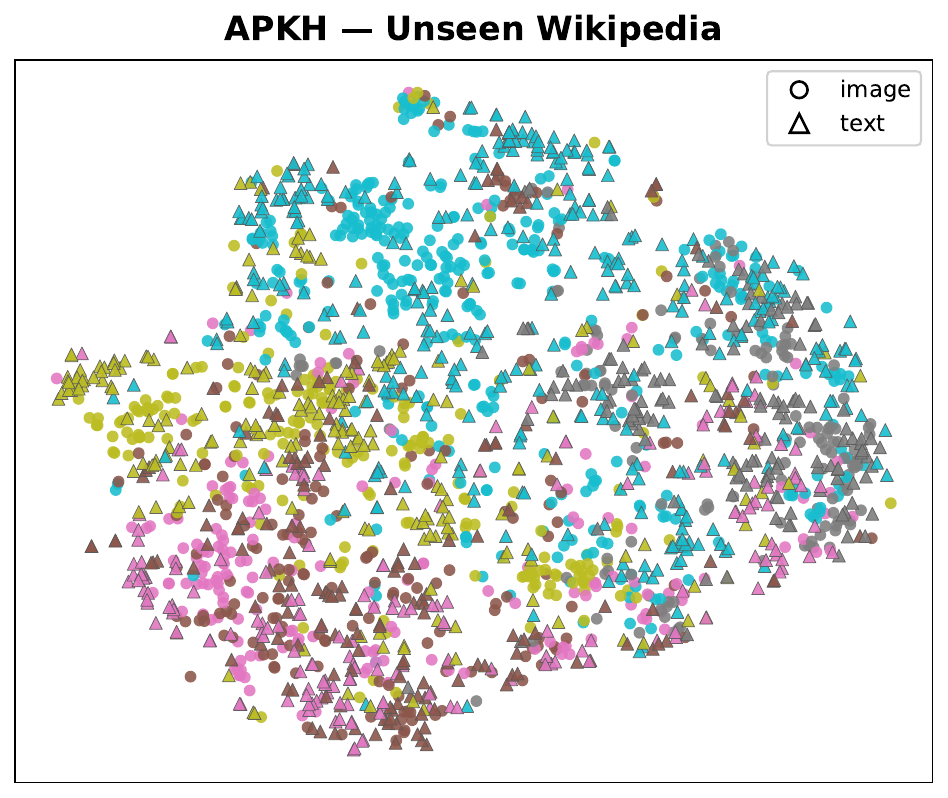}
    \caption{The t-SNE visualizations of the relaxed hash codes on Pascal Sentence and Wikipedia datasets. The plots compare the feature distributions of CAGAN, CIRH, GNAH, and the proposed APKH across seen and unseen categories. Different colors denote distinct semantic classes.}
    \label{fig:tsneResults}
\end{figure*}

\subsection{Cross-Dataset Transfer}
Table~\ref{tab:cross} details the performance for cross-dataset transfer experiments under the 40-pair, 64-bit configuration. In this evaluation, we use the same models trained on the ``seen" split of a source dataset as before and directly test them on the ``seen" splits of different target datasets to assess their robustness against distribution shifts and their domain generalization capabilities. The results indicate that APKH consistently exhibits superior transferability compared to existing unsupervised methods across all source-target combinations.  
For example, when using MIR Flickr as the source, APKH achieves an average mAP of 0.295 across all target datasets, significantly outperforming CIRH (0.235) and CAGAN (0.221). A similar trend is observed with the Pascal Sentence source dataset, where APKH reaches an mAP of 0.592 on the MIR target, marking a substantial improvement over the 0.465 achieved by CIRH and 0.440 by CAGAN. Furthermore, on the Wikipedia dataset, APKH maintains its performance lead with an overall average mAP of 0.463, notably surpassing its predecessor GNAH (0.433). It is worth noting that while the strict statistical regularization in APKH may slightly constrain its peak fitting performance on seen categories of Wikipedia, its superior generalization to target datasets outpaces existing methods. This underscores that APKH prioritizes learning robust, modality-invariant semantics over merely memorizing dataset-specific patterns.
These results demonstrate APKH successfully mitigates the performance degradation typically associated with cross-domain retrieval, proving its effectiveness for practical applications involving diverse and shifting data distributions.

\subsection{Visualization}
To intuitively illustrate the effectiveness of the proposed framework in bridging the modality gap and establishing a discriminative Hamming space, we visualize the generated hash codes using t-SNE. Fig.~\ref{fig:tsneResults} presents the feature distributions for representative baselines (CAGAN, CIRH), our preliminary method (GNAH), and the proposed APKH on the Pascal Sentence and Wikipedia datasets. The visualizations capture both seen and unseen categories, with distinct colors denoting different semantic classes and different markers representing the image and text modalities.

\textbf{Alignment in Seen Categories.} When looking at baseline methods like CAGAN and CIRH, a basic clustering pattern appears for seen categories, but the dividing lines between different semantic groups remain blurred. Additionally, text and image data points belonging to the same semantic class do not align closely, revealing an ongoing gap between the modalities. Although GNAH manages to gather semantics reasonably well, its image and text representations remain noticeably separated, which is an indicator of overfitting. In contrast, APKH shows a much stronger capacity to merge heterogeneous data into a single, cohesive space.

\textbf{Robustness in Unseen Categories.} A more significant difference becomes apparent in the more challenging unseen categories, where the limitations of the baseline methods are characterized by a widened modality gap. When encountering novel semantics, the baseline models fail to align heterogeneous data, causing their feature mappings to disintegrate into disjointed and fragmented distributions where image and text representations remain isolated from one another. APKH, conversely, effectively bridges this modality gap, maintaining clear class separation and structural soundness in the unseen domain. By utilizing CAKM and KSCA, it smooths out the semantic manifold instead of simply memorizing isolated data points. Consequently, the network can gracefully apply modality-independent semantic connections to tightly align image and text representations of new classes. This offers strong visual proof that APKH overcomes the severe modality gap seen in baselines, backing up the quantitative retrieval results and highlighting the framework's outstanding generalization capabilities under data scarcity.

\subsection{Ablation Study}

\textbf{Number of Attributes ($K$).} As shown in Fig.~\ref{fig:attribute}, the empirical evaluation of the attribute-prompted kernel hashing framework reveals that retrieval performance improves significantly as the number of attributes scales from 10 to 100. Beyond this initial growth, performance gains begin to level off after $K$ reaches 100, suggesting that this threshold is sufficient to capture the essential semantic information within the datasets. We employ all 620 attributes in the other parts of our experiments.

\textbf{Number of Context Tokens ($L$).} As shown in Fig.~\ref{fig:ctx}, we investigate the performance sensitivity w.r.t. the number of context tokens $L$ across the values of $\{1, 2, 4, 8, 16\}$. It can be observed that the cross-modal retrieval mAP increases initially and quickly reaches a stable plateau when $L \ge 2$ on both seen and unseen categories. This behavior indicates that a relatively short context length is already capable of providing sufficient parametric flexibility to steer the semantic attributes for modality alignment. Following the conventional configurations in the prompt tuning literature~\cite{yao2024tcp,li2025class}, we empirically select $L=4$ as the default token length for all other experiments to maintain an optimal balance between training efficiency and semantic expressiveness.

\begin{table*}[]
\renewcommand{\arraystretch}{1.1} 
\centering
\caption{Performance of APKH variants at 64 bits with 40 training sample pairs. \textbf{Bold} denotes the best result, \underline{underline} denotes the second-best result.}
\label{tab:ablation}
\footnotesize
\begin{tabular*}{0.99\textwidth}{@{\extracolsep{\fill}} l|ccccc|ccccc|ccccc @{}}
\toprule
\multirow{2}{*}{\textbf{Method}} & \multicolumn{5}{c|}{\textbf{Seen}} & \multicolumn{5}{c|}{\textbf{Unseen}} & \multicolumn{5}{c}{\textbf{Average}} \\
\cmidrule{2-6} \cmidrule{7-11} \cmidrule{12-16}
& MIR & NUS & PAS & WIKI & \textbf{AVG} & MIR & NUS & PAS & WIKI & \textbf{AVG} & MIR & NUS & PAS & WIKI & \textbf{AVG} \\
\midrule
APKH-A & 0.850 & 0.370 & 0.263 & 0.352 & 0.459 & 0.717 & 0.244 & 0.201 & 0.244 & 0.351 & 0.784 & 0.307 & 0.232 & 0.298 & 0.405 \\
APKH-B & 0.866 & 0.414 & \underline{0.559} & \textbf{0.411} & 0.563 & 0.713 & 0.284 & \underline{0.317} & \underline{0.345} & \underline{0.415} & 0.789 & 0.349 & \underline{0.438} & \textbf{0.378} & \underline{0.489} \\
APKH-C & \underline{0.868} & \underline{0.513} & 0.530 & 0.399 & \underline{0.577} & \textbf{0.719} & \underline{0.285} & 0.270 & 0.299 & 0.393 & \underline{0.794} & \underline{0.399} & 0.400 & 0.349 & 0.485 \\
APKH   & \textbf{0.874} & \textbf{0.525} & \textbf{0.592} & \underline{0.405} & \textbf{0.599} & \textbf{0.719} & \textbf{0.301} & \textbf{0.335} & \textbf{0.346} & \textbf{0.425} & \textbf{0.797} & \textbf{0.413} & \textbf{0.463} & \underline{0.376} & \textbf{0.512} \\
\bottomrule
\end{tabular*}
\end{table*}

\textbf{Number of Virtual Samples ($M$).} As shown in Fig.~\ref{fig:sampling}, we investigate the sensitivity of the proposed method with respect to the number of virtual samples $M \in \{1, 2, 5, 10, 20\}$. Interestingly, we observe that the retrieval performance remains relatively stable across different choices of $M$. This phenomenon can be attributed to our dynamic optimization strategy: since the virtual hash codes are re-sampled from the kernel distribution $\mathcal{P}_{i}$ at each training epoch, the continuous semantic manifold is already implicitly and thoroughly explored over the entire training course via this time-dynamic regularizer. Considering that a larger $M$ inevitably expands the contrastive denominator pool and increases the computational latency per epoch, we empirically set $M = 5$ to guarantee training efficiency without sacrificing the retrieval accuracy.

\textbf{Impact of $\alpha$.} As shown in Fig.~\ref{fig:alpha}, we investigate the impact of the hyperparameter $\alpha$ on cross-modal retrieval performance across different bit widths. The results are averaged between seen and unseen classes across all datasets. Overall, introducing a relatively small $\alpha$ ($\alpha<=1$) brings steady improvements in average performance across all bit widths compared to the baseline without binarization loss ($\alpha=0$). However, when $\alpha$ is set too large (e.g., $\alpha=100$), the retrieval performance experiences a severe collapse across all configurations. Furthermore, the onset of this negative effect is sensitive to the hash code length: the smaller the bit width, the earlier the performance degradation occurs. To prevent this performance degradation, we align with our preliminary work~\cite{11084616} and conservatively set $\alpha = 1$ across all bit lengths.

\begin{figure}[]
    \centering
    \includegraphics[width=0.45\textwidth]{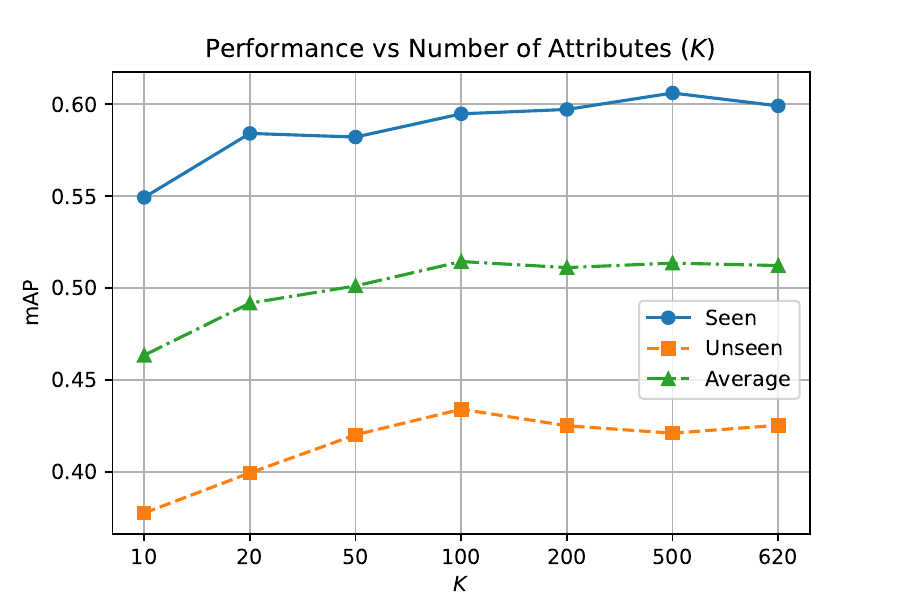}
    \caption{Impact of the number of attributes ($K$) on cross-modal retrieval performance. The results are averaged across all datasets.}
    \label{fig:attribute}
\end{figure}
\begin{figure}[]
    \centering
    \includegraphics[width=0.45\textwidth]{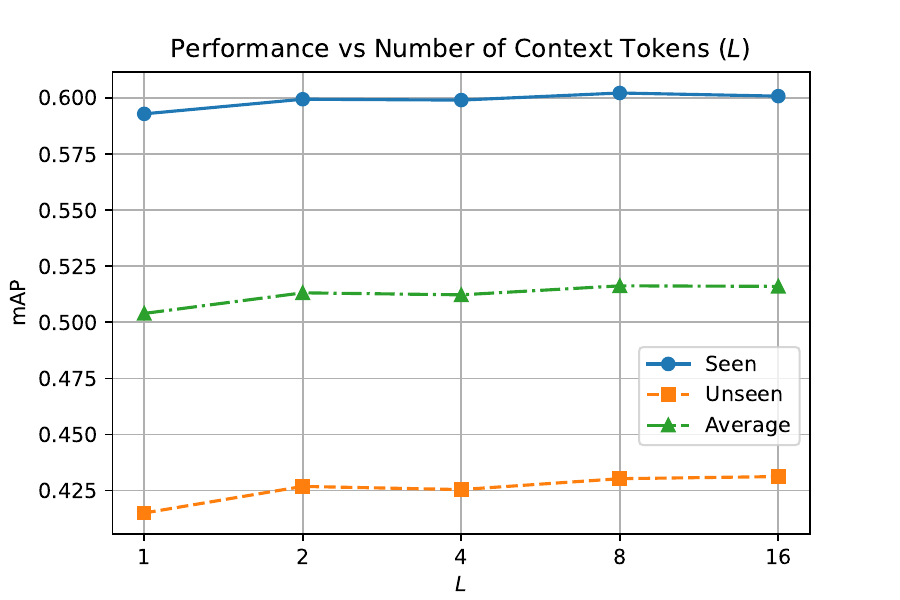}
    \caption{Impact of the number of context tokens ($L$) on cross-modal retrieval performance. The results are averaged across all datasets.}
    \label{fig:ctx}
\end{figure}

\begin{figure}[]
    \centering
    \includegraphics[width=0.45\textwidth]{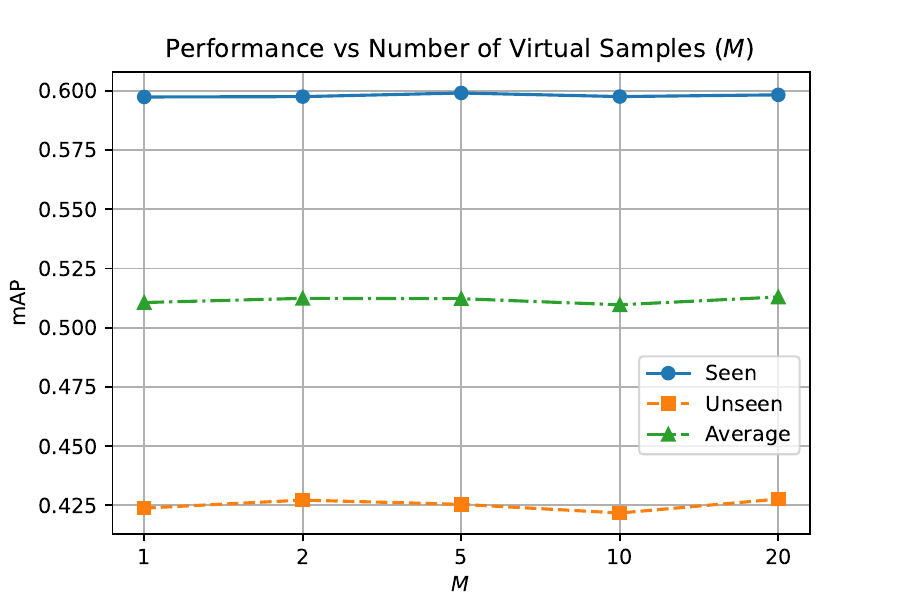}
    \caption{Impact of the number of virtual samples ($M$) on cross-modal retrieval performance. The results are averaged across all datasets.}
    \label{fig:sampling}
\end{figure}
\begin{figure}[]
    \centering
    \includegraphics[width=0.45\textwidth]{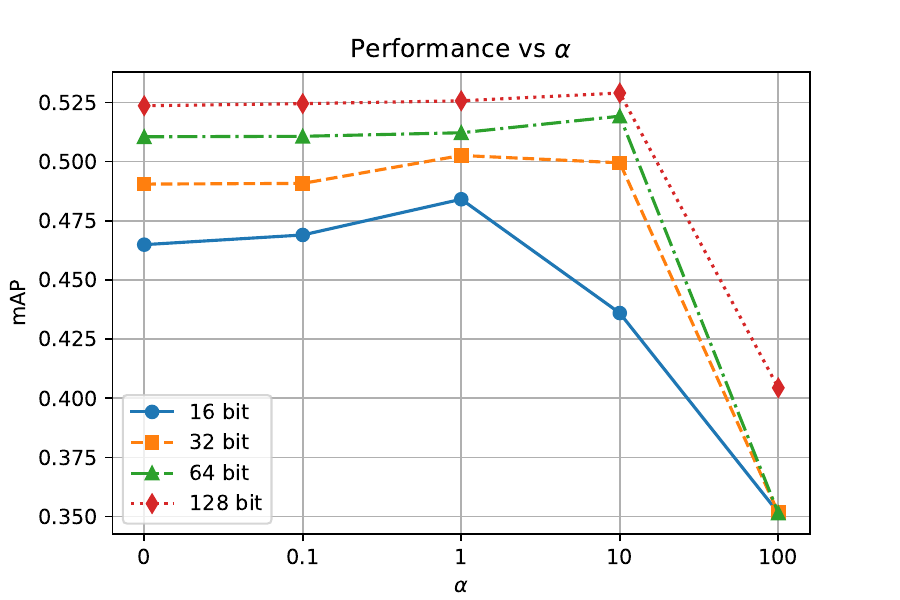}
    \caption{Impact of $\alpha$ on cross-modal retrieval performance. The results are averaged between seen and unseen categories across all datasets.}
    \label{fig:alpha}
\end{figure}

\textbf{Variants of APKH.} To investigate the impact of different components and explore alternative configurations, we construct several variants of APKH and show the performance comparisons in Table~\ref{tab:ablation}.
\textbf{APKH-A} generates kernels directly from attributes without using context optimization. This removal of context optimization results in the most severe performance degradation among all variants. Without learnable context prompts, the overall average mAP plummets from 0.512 (full APKH) to 0.405 (APKH-A). The performance drop is particularly notable in the unseen categories, where the average mAP falls from 0.425 to 0.351. This suggests that the raw attribute features from CLIP fall short for cross-modal retrieval tasks. Dynamically optimizing semantic attributes via prompt learning is critical for capturing modality-invariant semantics and bridging the modality gap.
\textbf{APKH-B} replaces the proposed KSCA with a standard contrastive loss. The overall average mAP drops to 0.489. This decline demonstrates the limitations of traditional discrete point-to-point matching. By modeling limited paired data as continuous distributions, the proposed KSCA effectively smooths the semantic gap and acts as a regularizer against overfitting.
\textbf{APKH-C} replaces the unified hash function with a conventional dual MLP hash function architecture. This modification reduces the overall average mAP to 0.485. Furthermore, its performance on unseen categories drops significantly to 0.393 (compared to the full model's 0.425). This confirms that mapping heterogeneous features into a shared set of kernel responses and processing them through a single unified network yields a more robust and generalizable Hamming space than relying on modality-specific projectors.

\section{Conclusion}
This paper presents Attribute-Prompted Kernel Hashing (APKH), a novel and correspondence-efficient unsupervised cross-modal retrieval framework designed to overcome the heavy reliance on large-scale image-text pairs and the poor zero-shot generalization of existing methods. To achieve this, APKH introduces two core modules: Context-optimized Attribute Kernel Mapping (CAKM), which leverages prompt learning and hyperspherical RBF kernel mapping to dynamically capture modality-invariant semantics and establish a robust alignment space, and Kernel-Smoothed Contrastive Alignment (KSCA), which explicitly smooths the modality gap by modeling limited paired data as continuous distributions rather than discrete points. By transitioning from discrete matching to continuous distribution alignment, KSCA acts as a powerful statistical regularizer that successfully mitigates overfitting to sparse pairwise correlations. Extensive experiments across four benchmark datasets demonstrate that APKH significantly outperforms state-of-the-art unsupervised cross-modal hashing methods, exhibiting improved generalization capability from seen to unseen categories under data-constrained scenarios.

\bibliographystyle{IEEEtran}
\bibliography{ref.bib}

\vfill

\end{document}